\RequirePackage{fix-cm}
\documentclass[smallextended]{svjour3_arXiv}       

\smartqed  
\usepackage{graphicx}
\usepackage{color}
\usepackage[centertags]{amsmath}
\usepackage{amssymb}
\usepackage{amsrefs}
\usepackage{epsfig}
\usepackage{float}
\usepackage{newlfont}

\newtheorem{prop}[theorem]{Proposition}
\newtheorem{cor}[theorem]{Corollary}
\newtheorem{lem}[theorem]{Lemma}

\newcommand{\RR}{\mathbb{R}}
\newcommand{\C}{\mathbb{C}}
\newcommand{\e}{\mathrm{e}}
\renewcommand{\l}{\ell}
\newcommand{\p}{\partial}
\newcommand{\Hv}{\vec{\mathrm{H}}}
\newcommand{\cf}{\bar{f}}
\newcommand{\cp}{\bar{p}}
\newcommand{\al}{\alpha}
\newcommand{\be}{\beta}
\newcommand{\ga}{\gamma}
\newcommand{\Del}{\Delta}
\newcommand{\Si}{\Sigma}
\newcommand{\si}{\sigma}
\newcommand{\om}{\omega}
\newcommand{\Ga}{\Gamma}
\newcommand{\Up}{\Upsilon}
\newcommand{\cov}{\nabla}
\newcommand{\pa}[2]{\dfrac{\partial #1}{\partial #2}}
\newcommand{\ppa}[2]{\dfrac{\partial^{2} #1}{\partial #2^{2}}}
\newcommand{\glap}[1]{\Box_{#1}}
\newcommand{\divx}[1]{\nabla_{#1} \cdot}
\newcommand{\Ric}{\operatorname{Ric}}
\newcommand{\II}{\operatorname{\text{I\hspace{-1 pt}I}}}
\newcommand{\pnth}[1]{\left( #1 \right)}
\newcommand{\abs}[1]{\left| #1 \right|}
\newcommand{\R}{\mathbf{R}}
\newcommand{\vep}{\varepsilon}
\newcommand{\brkt}[1]{\left[ #1 \right]}
\newcommand{\LHopital}{L'H\^{o}pital}
\newcommand{\Diag}{\operatorname{diag}}

\begin{document}

\title{A Survey of Spherically Symmetric Spacetimes
}


\author{Alan R.\ Parry 
}

\authorrunning{A.\ R.\ Parry} 

\institute{Alan R.\ Parry \at
              Department of Mathematics, University of Connecticut, 196 Auditorium Road, Unit 3009, Storrs, CT 06269-3009, USA \\
              Tel.: (860) 486-2362 \\
              Fax: (860) 486-4238 \\
              \email{alan.parry@uconn.edu}           
}


\maketitle

\begin{abstract}
We survey many of the important properties of spherically symmetric spacetimes as follows.  We present several different ways of describing a spherically symmetric spacetime and the resulting metrics.  We then focus our discussion on an especially useful form of the metric of a spherically symmetric spacetime in polar-areal coordinates and its properties.  In particular, we show how the metric component functions chosen are extremely compatible with notions in Newtonian mechanics.  We also show the monotonicity of the Hawking mass in these coordinates.  As an example, we discuss how these coordinates and the metric can be used to solve the spherically symmetric Einstein-Klein-Gordon equations.  We conclude with a brief mention of some applications of these properties.
\end{abstract}

\section{Introduction}

Spherically symmetric spacetimes are an important case in the study of general relativity for a number of reasons.  Foremost among them is that it is often a good starting point in the study of a problem in general relativity.  For example, one of the first projects undergone in general relativity was to compute nontrivial spherically symmetric spacetimes that are exact solutions of the Einstein equation.  This resulted in the discovery of the Schwarzschild spacetime, which is by far the most important spherically symmetric solution to date, and later Birkhoff's theorem about Ricci flat or vacuum spherically symmetric spacetimes as well as some generalizations of Birkhoff's theorem \cites{Wald84, Jebsen21, Birkhoff23, Bronnikov95}.  Spherically symmetric spacetimes also create a situation where the dynamics of the system are less complicated by effectively reducing a 4-dimensional solution to a 2-dimensional one.  This accessibility makes using spherically symmetric spacetimes all the more attractive as a starting point.  Finally, while Birkhoff's theorem classifies all vacuum spherically symmetric spacetimes, there are still some nonvacuum spherically symmetric spacetimes that are interesting as well, both from a physical and mathematical standpoint.  These include describing spherically symmetric stars and perfect fluids via the Tolman-Oppenheimer-Volkoff solutions and other methods \cites{Opp39,Tolman39,Heinzle03}, dwarf spheroidal galaxies, which are dominated by their dark matter halos and closely approximated by spherical symmetry \cite{Mash06}, and spherically symmetric scalar fields, particularly in the context of dark matter \cites{MSBS,Seidel90,Seidel98,Lai07,Gleiser89,Hawley00}.

As spherically symmetric spacetimes are still of interest and there is a great deal of information about them scattered throughout the literature, it would be useful to have a brief survey of many of the important results about spherically symmetric spacetimes.  The purpose of this paper is to collect a great deal of useful information about spherically symmetric spacetimes and present it in a brief organized way.  This includes, first, a discussion of the many possible forms of a generic metric of a spherically symmetric spacetime along with some advantages and disadvantages of each choice.  Most of these metrics can be described well within the established framework of numerical relativity and as such, we will use that framework to discuss them.  Second, to facilitate some of the discussion of spherical symmetry as well as to point out the particular usefulness of a certain choice of coordinates, we will discuss in detail what we refer to as the Newtonian-Compatible metric collecting some well known results about spherically symmetric spacetimes in terms of this metric.  Following the discussion of this metric, we will apply some of these results to a particular example, namely, a spherically symmetric scalar field.  For brevity and comprehensibility, we will collect most of the proofs of these results in an appendix.  Even though none of the results presented in this article are particularly new or unknown and most, if not all, have been previously published, it is the hope of the author that many will find this brief and dense collection of these results useful.


Before we get started, we should explain a notational convention about subscripts that we use.  In particular, we do not follow the commonly used abstract index notation unless explicitly stated.  Instead, we will generally use the following rules.

For functions, like $M(t,r)$, unless the context makes it clear, subscripts will denote partial differentiation with respect to the coordinate the subscript specifies, so that
\begin{subequations}
  \begin{equation}
    M_{r} = \pa{}{r} M(t,r) \qquad \text{and} \qquad M_{rt} = \pa{}{t}\pnth{\pa{}{r}M(t,r)}.
  \end{equation}

Vector fields will typically be uniquely named and as such any subscript they have is simply part of that name and usually denotes that that vector is in the coordinate direction of the subscript.  For example, $\p_{t}$ is the coordinate vector field corresponding to the $t$-coordinate and later we will define $\nu_{t}$ as the unit vector field in the $t$-coordinate direction.  We will also use a similar practice with one-forms where appropriate.  Given this practice, we will explicitly define the vector fields and one forms we use.

For tensors, it is often more convenient to have the subindices denote the slots of the tensor.  As such, we will never give a tensor (other than a function, vector field, or one-form) a subscripted name.  When subscripts appear on a tensor or on the Christoffel symbols $\Ga$, they will be a letter corresponding to a coordinate and the subscript will denote plugging in that particular coordinate vector field into that slot.  Similarly, a superscript on a tensor will denote plugging in that particular coordinate one-form into that slot.  Thus raising and lowering indicies is done in the familiar way of appropriately contracting the tensor with the metric or its inverse.  We will also adopt the abstract index notation convention of using a comma between subscripts to denote partial differentiation of a tensor component and a semi-colon between subscripts to denote covariant differentiation of a tensor in the direction indicated by the subscript.  Thus, for example,
  \begin{align}
    g_{tr} &= g(\p_{t},\p_{r}), \\
    \R_{tr\theta}^{\ \ \ \, \varphi} &= \R(\p_{t},\p_{r},\p_{\theta},d\varphi), \\
    T_{tr,\theta} &= \pa{}{\theta}\pnth{T(\p_{t},\p_{r})}, \\
    T_{tr;\theta} &= (\cov_{\theta}T)(\p_{t},\p_{r}) = (\cov T)(\p_{t},\p_{r},\p_{\theta}).
  \end{align}
\end{subequations}

\section{Metrics for a Spherically Symmetric Spacetime}

The study of numerical relativity is devoted to devising ways of evolving the Einstein equation in order to solve for the components of the spacetime metric in different situations of interest as well as actually conducting such numerical experiments and comparing them to real data.  This evolution usually takes place in a spacetime which is foliated by $t=constant$ spacelike hypersurfaces.  The common method in numerical relativity is to decouple the time component from the space components into what is commonly called the (3+1)-formalism of general relativity \cites{Chop98,Gour12,Alcu08}.  Several formulations of this formalism are usually used in practice, such as the Baumgarte-Shapiro-Shibata-Nakamura (BSSN) formulation \cites{Shibata95,Baumgarte99}, the Z4 formulation \cite{Bona03}, or the fully constrained formulation (FCF) \cite{Bonazzola04} among others, but the formulation described below is sufficient for our present survey of spherically symmetric spacetimes.

The framework for this formalism is, as stated before, a spacetime, $N$, foliated by $t=constant$ spacelike hypersurfaces described by a Riemannian 3-metric $\ga$ which may change with time.  Note that such a foliation is possible for any globally hyperbolic spacetime \cites{Gour12,Wald84}.  Consider a coordinate chart on $U \subseteq N$, $\{t,x^{1},x^{2},x^{3}\}$, where $\p_{t}$ is timelike and the $\p_{x^{j}}$ are all spacelike.  Now consider an observer starting on the $t=t_{0}$ hypersurface at the coordinate $(t_{0},\vec{x}^{j}_{0})$.  This observer then travels to another infinitesimally close hypersurface $t = t_{0}+dt$ to the coordinate $(t_{0} + dt, \vec{x}^{j}_{0} + d\vec{x}^{j})$ as in Figure \ref{3plus1_fig}.

\begin{figure*}[h]

\begin{center}
  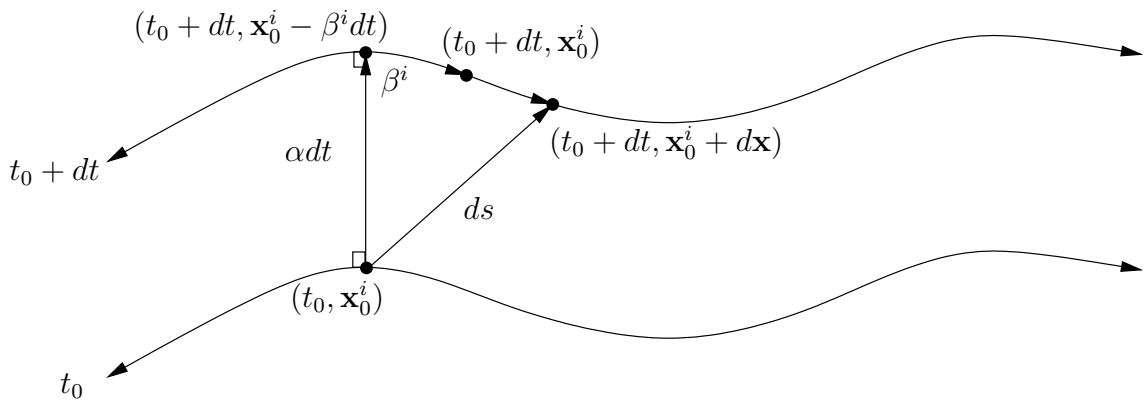
\end{center}

\caption{Infinitesimal distance in a $t=constant$ hypersurface foliated spacetime.}

\label{3plus1_fig}

\end{figure*}

The observer has now traveled an infinitesimal distance $ds$.  We can measure the ``square'' of this infinitesimal distance, $ds^{2}$, using the analogue of Pythagorean's theorem and this will give us the line element form of the metric.  If another observer travels normal to the hypersurface from $(t_{0},\vec{x}^{j}_{0})$ to the hypersurface $t=t_{0}+dt$, since the normal direction is not necessarily the same direction as the $t$ coordinate direction, it will arrive at the coordinate $(t_{0} + dt, \vec{x}^{j}_{0} - \beta^{j}dt)$.  We call the 3-vector field, $\beta$, the shift vector because it measures the spacelike shift of the coordinates while traveling normally.  Note that the components of $\beta$ can vary with all the coordinates.  This normal observer, having traveled in a timelike direction, has experienced some proper time $d\tau$, which is some multiple of the change in time coordinate, that is,
\begin{equation}
  d\tau = \al dt.
\end{equation}
Hence the length of its normal movement from one surface to the other is $\al dt$.  The value of $\al$ can vary with all of the coordinates making it a function on the manifold.  This function is called the lapse function since it measures the lapse in proper time compared to coordinate time.  To get the length in the spatial direction between where the normal observer ended up and where the original observer did, we need only to use the metric on the hypersurfaces and the difference of the two space coordinates.  This difference, for each $i$, takes the form
\begin{equation}
  (x^{j}_{0} + dx^{j}) - (x^{j}_{0} - \beta^{j}dt) = dx^{j} + \beta^{j}dt
\end{equation}
and so the length squared of the spatial movement will be
\begin{equation}
  \ga_{jk}(dx^{j} + \beta^{j}dt)(dx^{k} + \beta^{k}dt)
\end{equation}
where we have implemented the Einstein summation convention.  Then using the generalized version of Pythagorean's theorem and recalling that $t$ is a timelike direction we get that
\begin{equation}\label{3plus1_metric}
  ds^{2} = -\al^{2}\, dt^{2} + \ga_{jk}(dx^{j} + \beta^{j} dt)(dx^{k} + \beta^{k}dt),
\end{equation}
which is the line element of the metric.  Note that the lapse function $\al$ and the shift vector $\be$, both of which together are usually referred to as the guage variables, completely determine the foliation of the spacetime by the $t=constant$ spacelike hypersurfaces.  Equation (\ref{3plus1_metric}) is then the most general form of the metric for any foliated spacetime, that is, all metrics of a foliated spacetime can be written in this form.  Note that we will often write $g$ instead of $ds^{2}$ to refer interchangeably to the metric and the line element.

From this metric, and a choice of slicing condition, a complete system of partial differential equations that evolve the Einstein equation can be constructed.  This system is often referred to as the ADM formulation of general relativity in reference to the authors of the paper in which it was first introduced \cite{ADM08}.  It involves evolution equations of both the metric and the extrinsic curvature of the $t=constant$ hypersurfaces \cites{ADM08,Chop98,Bona02}.  These equations are very commonly used in numerical relativity, but we find they overcomplicate the situation in spherical symmetry, which is why we have elected not to use this formulation of general relativity directly.

If we know more about the spacetime in question, we will be able to determine more of the components of the metric.  In a spherically symmetric spacetime, with coordinates, $r,\theta,\varphi$, chosen so that $\theta$ and $\varphi$ are the polar-angular coordinates on the hypersurface, the shift vector must be completely radial so that the metric remains invariant under rotations.  In this case, we will denote the radial component of the shift vector by simply $\beta$.  Moreover, the $3$-metric can be written as
\begin{equation}
  \ga = \ga_{rr}\, dr^2 + \ga_{\theta\theta}\, d\si^{2}.
\end{equation}
where $d\si^{2} = d\theta^{2} + \sin^{2}\theta\, d\varphi^{2}$ is the standard metric on the unit sphere.  This implies that the most general spherically symmetric metric can be written in the form
\begin{align}\label{3plus1_ssmetric_a}
  \notag g &= -\al^{2}\, dt^{2} + \ga_{rr}(dr + \beta\, dt)(dr + \beta\, dt) + \ga_{\theta\theta}\, d\si^{2} \\
  &= -\pnth{\al^{2} - \ga_{rr}\beta^{2}}\, dt^{2} + \ga_{rr}\beta(dr\, dt + dt\, dr) + \ga_{rr}\, dr^{2} + \ga_{\theta\theta}\, d\si^{2}.
\end{align}
Note that all the metric component functions can only depend on $t$ and $r$ due to spherical symmetry.  For convenience, we will define two positive functions $a(t,r)$ and $q(t,r)$ so that
\begin{equation}
  a(t,r)^{2} = \ga_{rr} \qquad \text{and} \qquad q(t,r)^{2} = \ga_{\theta\theta}.
\end{equation}
This assignment is possibile because $\p_{r}$ and $\p_{\theta}$ are both spacelike and hence $\ga_{rr}$ and $\ga_{\theta\theta}$ are both positive functions.  Then we can rewrite (\ref{3plus1_ssmetric_a}) as
\begin{equation}\label{3plus1_ssmetric_b}
  g = -\pnth{\al^{2} - a^{2}\beta^2}\, dt^{2} + a^2\beta(dr\, dt + dt\, dr) + a^2\, dr^{2} + q^2\, d\si^{2},
\end{equation}
where $\al, a, \beta$, and $q$ are functions of only $t$ and $r$.  It is important to note here that since we have not defined the coordinates $t$ and $r$ geometrically yet, there remain two degrees of freedom left in this metric. There are several different choices that can be made in this regard.  We mention the most common here, but there is a very useful and more extensive list of several choices that can be made in a more general setting in Gourgoulhon's recent book \cite{Gour12}.

There are three rather common slicing conditions that are often used in many settings, all of which place a condition on the lapse function $\al$.  These conditions are the maximal slicing, harmonic slicing, and geodesic slicing conditions.

Under the maximal slicing condition, one requires that each $t=constant$ hypersurface be a maximal hypersurface, that is, it has zero mean curvature,
\begin{equation}
  H=0.
\end{equation}
The metric stays of the form in equation (\ref{3plus1_ssmetric_a}), but since the evolution equations in the ADM formulation evolve the components of the second fundamental form and $H$ is the trace of the second fundamental form, this places a constraint on some of the evolution variables, which can be used to simplify the evolution equations \cites{Cord11,Gour12,Baum10}.  Note here that after making this choice, there remains one more degree of freedom which can be used to constrain the $r$ coordinate.

Harmonic slicing requires that the coordinate function $t$ be a harmonic function under the metric $g$.  That is, $\glap{g}t = 0$, where $\Box_{g}$ is the d'Alembertian or Laplacian operator with respect to the metric $g$.  This is often accompanied with the condition that the hypersurfaces remain orthogonal to the time direction, which uses the remaining degree of freedom, and indeed some refer to both of these choices together as harmonic slicing.  In the case that both conditions are satisfied, this would yield a metric of the form
\begin{equation}
  g = -\al^{2}\, dt^{2} + a^{2}\, dr^{2} + q^{2}\, d\si^{2}
\end{equation}
with the added condition $\glap{g}t = 0$, which can be used to compute an evolution equation for the lapse function $\al$ \cite{Bona92}.

Geodesic slicing requires that movement along the curve $\xi = (t,0,0,0)$, which is given in our coordinates, be geodesic, that is, coordinate observer worldlines are geodesics.  This requirement is satisfied by choosing
\begin{equation}
  \al = constant \qquad \text{and} \qquad \be=0.
\end{equation}
However, the most reasonable choice for the constant is 1, since the condition $\al=1$ on the lapse function has the added implication that a normal observer's proper time is the same as coordinate time and in fact, since $\be =0$, normal observers are coordinate observers \cites{Gour12, Baum10, Alcu00}.  This results in a metric of the form
\begin{equation}
  g = -dt^{2} + a^2\, dr^{2} + q^2\, d\si^{2}.
\end{equation}
While this choice seems very attractive at first, since it either eliminates or greatly simplifies the evolution equations, it does have a tendency to develop coordinate singularities when evolved in time and also has considerable trouble dealing with real singularities arising, for example, from the collapse of a star to a black hole \cite{Baum10}.  As such, this choice of slicing must be used with caution.  Note also that this requirement is sometimes reduced to simply $\al = 1$ without necessarily requiring the shift parameter or vector to vanish.

We can alternatively use the degrees of freedom to make choices concerning the $r$-coordinate.  We will mention a few such choices here that are standard in the study of spherically symmetric spacetimes.

The first is the choice of normal slicing.  That is, choose the $t$ coordinate so that the $\p_{t}$ vector field is always normal to the hypersurfaces.  This choice makes all normal observers coordinate observers as well.  This is equivalent to choosing the shift parameter or vector to be identically 0 and results in a metric of the form
\begin{equation}\label{normal-slice}
  g = -\al^{2}\, dt^{2} + a^{2}\, dr^{2} + q^{2}\, d\si^{2}.
\end{equation}
This choice has already been mentioned above as it is often coupled with harmonic or geodesic slicing conditions \cite{Gour12}.  Since the coordinate vector fields on the hypersurfaces were already orthogonal, this results, as seen above, in a diagonal metric.

The next choice is to require that the metric on the hypersurfaces to be conformal to the flat metric.  This amounts to choosing the function $q$ in (\ref{3plus1_ssmetric_b}) to satisfy $q = ra$, which would make the metric become
\begin{equation}
  g = -\pnth{\al^{2} - a^{2}\beta^2}\, dt^{2} + a^2\beta(dr\, dt + dt\, dr) + a^2 \pnth{dr^{2} + r^2\, d\si^{2}}.
\end{equation}
This choice is referred to as isotropic coordinates.  It is often coupled with the normal slicing choice above, which uses both of the degrees of freedom and results in a metric of the form
\begin{equation}\label{isotrop}
  g = -\al^{2}\, dt^{2} + a^2 \pnth{dr^{2} + r^2\, d\si^{2}}.
\end{equation}
A common example of the normal-isotropic case is the Schwarzschild metric in isotropic coordinates \cite{Wald84}.

A form of the metric closely related to isotropic coordinates is the confomally flat condition.  Here again, we require that the metric on the hypersurfaces to be conformal to the flat metric on $\RR^{3}$, but this time we require the function $q$ in (\ref{3plus1_ssmetric_b}) to satisfy $q = r\psi^{2}$ to obtain the metric
\begin{equation}\label{CFC-cond}
  g = -\pnth{\al^{2} - \psi^{4}\beta^2}\, dt^{2} + \psi^{4}\beta(dr\, dt + dt\, dr) + \psi^{4} \pnth{dr^{2} + r^2\, d\si^{2}}.
\end{equation}
This choice is mathematically identical to the isotropic coordinate choice above; the only difference is that it makes the fact that the spatial portion of this metric is conformal to the flat metric on $\RR^{3}$ match the usual definition of conformally flat.  The conformally flat condition is also used in other instances besides spherical symmetry.  However, since this choice can always be made in spherical symmetry, we see that in spherical symmetry, the conformally flat condition is exact \cite{Cord11}.  Just like isotropic coordinates, the conformally flat condition is often coupled with the normal slicing choice, using up the other degree of freedom by setting $\be = 0$ and yielding the metric
\begin{equation}
  g = -\al^{2}\, dt^{2} + \psi^{4}\pnth{dr^{2} + r^{2}\, d\si^{2}}.
\end{equation}
It is also sometimes coupled with the maximal slicing condition above that the mean curvature vanishes.  In this case, the metric remains in the form of equation (\ref{CFC-cond}), but the equation $H=0$ simplifies the evolution equations just as before.

We should note that the conformally flat condition only imposes that the $3$-metric on the $t=constant$ hypersurfaces be conformally flat.  There is also a condition, called conformally flat spacetime coordinates, which require that the entire spacetime be conformally flat, that is, conformal to the Minkowski metric.  This condition in spherical symmetry has been recently discussed in detail by Gr{\o}n and Johannesen \cite{Gron13}.

Another choice on the $r$-coordinate is to choose it to be the ``quasiglobal'' coordinate.  Here we choose the $r$-coordinate so that 
\begin{equation}
  g_{tt}g_{rr} = -1.
\end{equation}
This is often coupled with the normal slicing condition above that $\be = 0$, which uses the remaining degree of freedom.  In this case, this condition becomes $\al a = -1$ and the metric takes the form
\begin{equation}
  g = -\al^{2}\, dt^{2} + \frac{1}{\al^{2}}\, dr^{2} + q^{2}\, d\si^{2}.
\end{equation}
This coordinate choice is useful when dealing with black holes and both sides of horizons.  It is discussed in much greater detail in the book by Bronnikov and Rubin \cite{Bronnikov13}.

Similar to choosing the $t$-coordinate to be a harmonic function of the metric, we can also require the radial coordinate to be a harmonic function of the metric.  This is called the harmonic radial coordinate.  In this case, we still have a metric of the form in equation (\ref{3plus1_ssmetric_b}), but we also require that $\glap{g} r = 0$.  As in many other cases, this is often coupled with the normal slicing condition that $\be = 0$, using the remaining degree of freedom and yielding a metric of the form in equation (\ref{normal-slice}) with the added condition that $\glap{g} r = 0$.  This choice is well suited to dealing with scalar fields, although we will use a different choice when discussing a scalar field later in the survey.  This choice is also described in detail in the book by Bronnikov and Rubin \cite{Bronnikov13}.

The next $r$-coordinate choice we present is called the ``tortoise'' coordinate.  This choice requires $r$ to be such that
\begin{equation}
  g_{tt} = -g_{rr}.
\end{equation}
This choice uses only one degree of freedom 
and can be realized by solving $g_{tt}=-g_{rr}$ for the metric in equation (\ref{3plus1_ssmetric_b}).  This amounts to rewriting $a^{2}$ as follows
\begin{equation}
  a^{2} = \frac{\al^{2}}{1 + \be^{2}},
\end{equation}
which yields a metric of the form
\begin{equation}
  g = - \pnth{\frac{\al^{2}}{1+\be^{2}}}\, dt^{2} + \pnth{\frac{\be \al^{2}}{1+\be^{2}}}(dr \, dt + dt \, dr) + \pnth{\frac{\al^{2}}{1+\be^{2}}}\, dr^{2} + q^{2}\, d\si^{2}.
\end{equation}
Note that both of the gauge variables $\al$ and $\be$ still appear in the metric and retain their meaning.  
If this choice is coupled with normal slicing (i.e. $\be=0$), using the remaining degree of freedom, we obtain a metric of the form
\begin{equation}
  g = -\al^{2} \, dt^{2} + \al^{2}\, dr^{2} + q^{2}\, d\si^{2}.
\end{equation}
This choice is closely related to null coordinates and the Eddington-Finkelstein coordinates described below.  Many of the advantages of this choice of coordinate are discussed in detail in the book by Bronnikov and Rubin, but they appear in other texts as well \cites{Bronnikov13,Baum10}.

The last choice we mention in the (3+1) framework is to give the coordinate $r$ geometric significance by choosing it to be the areal coordinate.  That is, choose the coordinate $r$ so that the area of each metric $2$-sphere on the hypersurface is exactly $4\pi r^{2}$.  This choice requires that $q = r$.  Additionally, it is almost always accompanied with the normal slicing choice above, again using both degrees of freedom.  This results in the polar-areal coordinates on a spherically symmetric spacetime and yields a metric of the form
\begin{equation}
  g = -\al^{2}\, dt^{2} + a^{2}\, dr^{2} + r^{2}\, d\si^{2}.
\end{equation}
This chart is probably the most familiar form of a general spherically symmetric metric, and is the chart most commonly used when introducing the Schwarzschild spacetime.  And rightly so, as it has some very clear advantages.  For one, the $r$-coordinate's role is analogous to its role in a flat spacetime.  Additionally, these coordinates give the Einstein curvature tensor a very simple form.  However, they may not be well suited to dealing with high gravitational fields as we would likely run into the same limitations that the polar areal metric has in describing the Schwarzschild spacetime inside the Schwarzschild radius.  This is not much of a problem, however, if one is interested in working in the low field limit such as describing objects on a galactic scale.

We present one more useful coordinate system and metric for a general spherically symmetric spacetime that does not fit into the (3+1)-formalism framework, nor does it depend on the ability to foliate the spacetime, but is useful from a theoretical standpoint nonetheless.  In these coordinates, we still choose to use the polar-angular coordinates $\theta$ and $\varphi$ to describe the rotations, but instead of separating the time and radial coordinates, we choose coordinates $u$ and $v$ such that $\p_{u}$ and $\p_{v}$ are nonparallel future pointing null vectors.  This coordinate system is descriptively called null coordinates.  In this case, the metric on the spacetime takes the form
\begin{equation}
  g = A\, (du \, dv + dv\, du) + Q^{2}\, d\si^{2}
\end{equation}
where $A$ and $Q$ are functions of only $u$ and $v$.  Note that there are no extra degrees of freedom with this metric as all coordinates have been well-defined.  While these coordinates are not very well suited to numerical evolutions, they can be very useful for theoretical discussions about the spherically symmetric spacetime. For example, in these coordinates, it is very straightforward to prove the monotonicity of the Hawking mass given the dominant energy condition.  Additionally, it seems to perform well when in the presence of high gravitational fields.  In the Schwarzschild case, these coordinates are known as the Kruskal coordinate system and is the system generally used to describe the region inside the Schwarzschild radius, although other equally useful coordinate systems exist to describe this region such as the Eddington-Finkelstein coordinates and the Painlev\'{e}-Gullstrand-Lema\^{i}tre coordinates \cites{Wald84,Martel01}.

We make one final note here about static spacetimes.  If the spherically symmetric spacetime is also known to be static, meaning that there exists a timelike killing vector field with orthogonal spacelike hypersurfaces, then we can automatically eliminate the cross term in the general metric (\ref{3plus1_ssmetric_b}) by selecting the time coordinate to be in the direction of the timelike killing vector field and choosing the remaining coordinates to be the general spherical coordinates on the orthogonal spacelike hypersurfaces.  This effectively sets $\be = 0$.  A survey article on static spherically symmetric spacetimes can be found in \cite{Deser05}.

All of these different coordinate choices and different forms of the metric on a spherically symmetric spacetime have advantages to them.  However, for the problem of numerically evolving a spacetime metric in a low gravitational field, we find that the polar-areal coordinates are the best suited due to the very simple system one gets from the Einstein equation.  As such, polar-areal coordinates is the coordinate system that we will use throughout this paper, but, in addition, we will introduce new variables that will give the metric a different form.  This new form of the metric will result in the added advantages that the new metric functions have very clear analogues in the Newtonian or low-field limit and the Einstein curvature tensor will become even more simplified.

\section{A Newtonian-Compatible Metric}

Let $N$ be a spherically symmetric, (3+1)-dimensional spacetime with a Lorentzian metric $g$ and be foliated by $t=constant$ spacelike hypersurfaces.  Choose the polar-areal coordinate system discussed above so that the $t$-coordinate direction is always normal to the hypersurfaces, the radial coordinate $r$ is the areal coordinate (that is, so that the metric sphere of metric radius $r$ always has a surface area of $4\pi r^2$), and $\theta$ and $\varphi$ are the usual polar-angular coordinates on the hypersurface.  In these coordinates, the metric $g$ has the line element form
\begin{equation} \label{metric-1}
  g = -\al(t,r)^{2}\, dt^{2} + a(t,r)^{2}\, dr^{2} + r^{2}\, d\si^{2}
\end{equation}
Now we define the functions $V(t,r)$ and $M(t,r)$ as follows,
\begin{align}\label{MV-eq}
  V &= \ln \al & M &= \frac{r}{2}\pnth{\frac{a^{2}-1}{a^{2}}}
\end{align}
and rewrite the metric (\ref{metric-1}) in terms of these functions to obtain,
\begin{equation} \label{metric-2}
  g = - \e^{2V}\, dt^{2} + \pnth{1 - \frac{2M}{r}}^{-1}\, dr^{2} + r^{2}\, d\si^{2}
\end{equation}
We will always be interested in the low-field limit, where $M \ll r$, which necessarily requires that $M(t,0)=0$ for all $t$.  Note that if $a$ is nonzero at $r=0$, then $M(t,0)=0$ follows directly from equation (\ref{MV-eq}).  However, if $a=0$ at $r=0$, then $M(t,0)$ could possibly be nonzero for some value of $t$.  Assuming $M \ll r$ is a sufficient condition to force $a$ to be nonzero at $r=0$.

\subsection{Compatibility with Newtonian Physics}

Here we compute several important properties about this metric and the physical interpretations of both $V$ and $M$.  To physically interpret this metric, we will introduce the Einstein equation, but first, we present a property that doesn't require the Einstein equation, the proof of which can be found in Appendix \ref{Proofs}.

\begin{prop} \label{mH-Prop}
  The function $M(t,r)$ is the spacetime Hawking mass of the metric sphere, $\Si_{t,r}$, for any given $t$ and $r$.
\end{prop}

This suggests that we should interpret $M$ as the mass of the system.  However, there is an even stronger reason to do so, which we will investigate later.  In order to prepare for that discussion, we need some preliminary results first about the relationship between the metric and its stress-energy tensor.  To begin, consider this metric as a solution to the Einstein equation
\begin{equation} \label{EinEq-1}
  G = 8\pi T,
\end{equation}
for some stress-energy tensor $T$.  To facilitate our discussion, we will compute the Einstein curvature tensor of this metric in certain directions.  To that end, define the following unit vector fields.
\begin{align} \label{norvec}
  \nu_{t} &=  \e^{-V}\, \p_{t} & \nu_{r} &= \sqrt{1 - \frac{2M}{r}} \, \p_{r} & \nu_{\theta} &= \frac{1}{r}\, \p_{\theta} & \nu_{\varphi} &= \frac{1}{r\sin \theta} \, \p_{\varphi}
\end{align}
Note that at every point, $p \in N$, except the coordinate singularities $r=0$ and $\theta=\pm \pi$ (that is, all points where all the vector fields above are well defined), these vector fields form an orthonormal basis of $T_{p}N$ and hence are a frame field.  The Einstein curvature tensor is defined as
\begin{equation}\label{ECT}
  G = \Ric - \frac{1}{2}Rg
\end{equation}
where $\Ric$ and $R$ are the Ricci curvature tensor and the scalar curvature of the spacetime respectively.  Since both the Ricci curvature tensor and its trace $R$ are present in this equation, we will need to know a few results about the Ricci curvature in these coordinates.  We have the following lemma and subsequent corollary.  While the proof of the corollary is short and is presented here, the proof of the lemma can be found in Appendix \ref{Proofs}.

\begin{lem} \label{RicLem-2}
  The only nonzero components of the Ricci curvature tensor in the $\nu_{\eta}$ basis and the scalar curvature are as follows.
    \begin{enumerate}
      \item\label{RicLem-2a} $\displaystyle \Ric(\nu_{t},\nu_{t}) = \pnth{V_{rr} + V_{r}^{2} + \frac{2V_{r}}{r}}\pnth{1-\frac{2M}{r}} + \frac{V_{r}}{r}\pnth{\frac{M}{r} - M_{r}} \\ \phantom{.} \hspace{7 eM} + \frac{V_{t}M_{t} - M_{tt}}{r \e^{2V}}\pnth{1-\frac{2M}{r}}^{-1} - \frac{3M_{t}^{2}}{r^{2} \e^{2V}}\pnth{1-\frac{2M}{r}}^{-2}$
      \item\label{RicLem-2b} $\displaystyle \Ric(\nu_{t},\nu_{r}) = \frac{2M_{t}}{r^{2} \e^{V}}\pnth{1 - \frac{2M}{r}}^{-1/2}$
      \item\label{RicLem-2c} $\displaystyle \Ric(\nu_{r},\nu_{r}) = -(V_{rr} + V_{r}^2)\pnth{1 - \frac{2M}{r}} - \pnth{\frac{2}{r^{2}} + \frac{V_{r}}{r}}\pnth{\frac{M}{r} - M_{r}} \\ \phantom{.} \hspace{7 eM} - \frac{V_{t}M_{t} - M_{tt}}{r \e^{2V}}\pnth{1-\frac{2M}{r}}^{-1} + \frac{3M_{t}^{2}}{r^2 \e^{2V}}\pnth{1-\frac{2M}{r}}^{-2}$
      \item\label{RicLem-2d} $\displaystyle \Ric(\nu_{\theta},\nu_{\theta}) = \Ric(\nu_{\varphi},\nu_{\varphi}) = \frac{1}{r^{2}}\pnth{\frac{M}{r} + M_{r}} - \frac{V_{r}}{r}\pnth{1 - \frac{2M}{r}}$
      \item\label{RicLem-2e} $\displaystyle R = -2\pnth{V_{rr} + V_{r}^{2} + \frac{2V_{r}}{r}}\pnth{1-\frac{2M}{r}} - \frac{2V_{r}}{r}\pnth{\frac{M}{r} - M_{r}} + \frac{4M_{r}}{r^2} \\ \phantom{.} \hspace{7 eM} - \frac{2(V_{t}M_{t} - M_{tt})}{r \e^{2V}}\pnth{1-\frac{2M}{r}}^{-1} + \frac{6M_{t}^{2}}{r^{2} \e^{2V}}\pnth{1-\frac{2M}{r}}^{-2}$
    \end{enumerate}
\end{lem}

\begin{cor} \label{GCor-1}
  The only nonzero components of the Einstein curvature tensor in the $\nu_{\eta}$ basis are as follows.
    \begin{enumerate}
      \item\label{GCor-1a} $\displaystyle G(\nu_{t},\nu_{t}) = \frac{2 M_{r}}{r^{2}}$
      \item\label{GCor-1b} $\displaystyle G(\nu_{t},\nu_{r}) = \frac{2M_{t}}{r^{2} \e^{V}}\pnth{1 - \frac{2M}{r}}^{-1/2}$
      \item\label{GCor-1c} $\displaystyle G(\nu_{r},\nu_{r}) = -\frac{2M}{r^{3}} + \frac{2V_{r}}{r}\pnth{1 - \frac{2M}{r}}$
      \item\label{GCor-1d} $\displaystyle G(\nu_{\theta},\nu_{\theta}) = G(\nu_{\varphi},\nu_{\varphi}) \\ \phantom{.} \hspace{3.6 eM} = \pnth{V_{rr} + V_{r}^{2} + \frac{V_{r}}{r}}\pnth{1-\frac{2M}{r}} + \pnth{\frac{M}{r} - M_{r}}\pnth{\frac{1}{r^{2}} + \frac{V_{r}}{r}} \\ \phantom{.} \hspace{7 eM} + \frac{V_{t}M_{t} - M_{tt}}{r \e^{2V}}\pnth{1 - \frac{2M}{r}}^{-1} - \frac{3M_{t}^{2}}{r^{2} \e^{2V}}\pnth{1-\frac{2M}{r}}^{-2}$
    \end{enumerate}
\end{cor}

\begin{proof} Since $\{\nu_{t},\nu_{r},\nu_{\theta},\nu_{\varphi}\}$ form an orthonormal basis everywhere,
  \begin{equation}
    G(\nu_{\alpha},\nu_{\beta}) = \Ric(\nu_{\alpha},\nu_{\beta}) - \frac{R}{2}g(\nu_{\alpha},\nu_{\beta}) = \Ric(\nu_{\alpha},\nu_{\beta}) - \frac{R}{2}\varepsilon_{\alpha \beta}
\end{equation}
where $\alpha, \beta \in \{t,r,\theta,\varphi\}$ and $\varepsilon_{\alpha\beta} = \Diag(-1,1,1,1)$.  Note that this implies that
  \begin{equation}
    G(\nu_{\varphi},\nu_{\varphi}) = G(\nu_{\theta},\nu_{\theta})
  \end{equation}
We then use Lemma \ref{RicLem-2} to substitute in the values for $\Ric$ and $R$. 
\end{proof}

Next, we define the function $\mu(t,r)$ to be the energy density of an observer at $p=(t,r)$ moving through the slices with 4-velocity $\nu_{t}$.  In the context of the stress energy tensor $T$ in equation (\ref{EinEq-1}), this means
\begin{equation} \label{ED-def}
  \mu = T(\nu_{t}, \nu_{t}).
\end{equation}
We now have enough information to prove the following proposition, which contains the promised stronger reason for interpreting $M$ as the mass inside each metric sphere.

\begin{prop} \label{ED-Prop-1}
  For a fixed $t$ and $r$, $M(t,r)$ is the flat integral of the energy density, $\mu(t,r)$, over the ball $E_{t,r}$ of radius $r$ at time $t$.
\end{prop}

\begin{proof}
  By the Einstein equation (\ref{EinEq-1}), equation (\ref{ED-def}), and Corollary \ref{GCor-1}, we have that
  \begin{align} \label{EDEq-1}
    \notag G(\nu_{t},\nu_{t}) &= 8\pi T(\nu_{t},\nu_{t}) \\
    M_{r} &= 4\pi r^{2}\mu
  \end{align}
  Since $\int_{\Si_{t,r}} dA = 4\pi r^{2}$, where $\Si_{t,r}$ is the sphere of radius $r$ at time $t$, we have then that for a fixed $t$ and $r$,
  \begin{align}
    \notag M(t,r) &= \int_{0}^{r} 4\pi s^{2} \mu(t,s)\, ds =\int_{0}^{r} \int_{\Si_{t,s}} \mu(t,s) \, dA \, ds \\
    &= \int_{0}^{r}\int_{0}^{2\pi} \int_{0}^{\pi} \mu(t,s) s^{2}\sin\theta \, d\theta\, d\varphi\, ds = \int_{E_{t,r}} \mu(t,s) \, dV_{0},
  \end{align}
  where $dV_{0} = s^{2}\sin\theta \, d\theta\, d\varphi\, ds$ is the flat volume form on the ball of radius $s$, and we have introduced $s$ as a dummy integrating variable in the ``$r$'' position.  Thus $M(t,r)$ is the flat volume integral of the energy density over the ball of radius $r$.
\end{proof}

Note that $M$ is not the integral of the energy density with respect to the metric's volume form, $dV = (1 - 2M/s)^{-1/2}\, s^{2}\sin\theta \, ds \, d\theta \, d\varphi$, but rather the following is true.
  \begin{equation}
    M(t,r) = \int_{E_{t,r}}\mu(t,s)\, \sqrt{1 - \frac{2M(t,s)}{s}} \, dV.
  \end{equation}
However, in the Newtonian limit, $M\ll r$, the above integral becomes approximately the integral of the energy density with respect to the metric's volume form over the ball $E_{t,r}$ of radius $r$.  Thus referring to $M(t,r)$ as the mass inside the metric sphere of radius $r$ at time $t$ not only makes sense from a geometrical point of view given the Hawking mass, but also from a physical point of view.

Furthermore, since the energy density is spherically symmetric and smooth at the origin, we must have $\mu_{r}(t,0) = 0$ for all $t$.  Thus for small $r$, $\mu$ is approximately constant and nonnegative.  In fact, in most typical cases, $\mu$ is strictly positive at $r = 0$.  In the case when $\mu(t,0) > 0$, the above integral yields for small $r$ that
\begin{equation}\label{massCub}
    M(t,r) = \int_{0}^{r} 4\pi s^{2}\mu(t,s)\, ds \approx \int_{0}^{r} 4\pi s^{2}\mu(t,0)\, ds = \frac{4\pi \mu(t,0)}{3}r^{3}.
\end{equation}
Thus the initial behavior of $M$ near $r=0$ is that of a cubic power function. 
In the case when $\mu(t,0) = 0$, a Taylor expansion of $\mu$ near $r=0$ yields that $\mu(t,r)$ is initially (i.e. near $r=0$) a quadratic power function.  Then a similar computation to equation (\ref{massCub}) yields that $M(t,r)$ is initially a $5^{\text{th}}$ power of $r$.  In either case, this implies that for all $t$
\begin{equation}\label{Mder}
    M(t,0) = 0, \qquad M_{r}(t,0) = 0, \qquad \text{and} \qquad M_{rr}(t,0) = 0.
\end{equation}
This fact will be useful in proving Proposition \ref{EinEq-prop1} where we will need it to apply \LHopital's rule.   

Next we define the function $P$ to be the pressure in the stress-energy tensor for an observer at $(t,r)$, that is, let
\begin{equation} \label{PR-def}
  P = T(\nu_{r},\nu_{r}).
\end{equation}
Then we can use Corollary \ref{GCor-1} and the Einstein equation (\ref{EinEq-1}) to prove the following proposition.

\begin{prop} \label{PR-Prop-1}
  In the Newtonian limit, where $P=0$ and $M\ll r$, we have that $\Del V = 4\pi \mu$, where $\Del$ is the flat Laplacian on $\RR^{3}$.
\end{prop}

\begin{proof}
  By the Einstein equation (\ref{EinEq-1}), equation (\ref{PR-def}), and Corollary \ref{GCor-1}, we have that
  \begin{align}\label{PREq-1}
    \notag G(\nu_{r},\nu_{r}) &= 8\pi T(\nu_{r},\nu_{r}) \\
    r^{2}V_{r} &= \pnth{1 - \frac{2M}{r}}^{-1}\pnth{M + 4\pi r^{3}P}.
  \end{align}
  Note that in the Newtonian limit, where $P=0$ and $M\ll r$, this equation is approximated by
  \begin{equation} \label{PotEq-1}
    r^{2}V_{r} = M.
  \end{equation}
  Also, the metric on the hypersurface, under these assumptions, is approximately the polar-areal metric on $\RR^{3}$.  Note that the Laplacian on $\RR^{3}$ of a spherically symmetric function $f$ is given by
  \begin{equation}
    \Del f= \ppa{f}{r} + \frac{2}{r}\pa{f}{r} = \frac{1}{r^{2}}\pa{}{r}\pnth{r^{2}\pa{f}{r}}
  \end{equation}
  Then applying the operator $\dfrac{1}{r^{2}}\pa{}{r}$ to (\ref{PotEq-1}) and using equation (\ref{EDEq-1}) yields
  \begin{equation}\label{PotEq-2}
    \Del V = \frac{1}{r^{2}}\pa{}{r}(r^{2}V_{r}) = 
   \frac{1}{r^{2}}M_{r} = 4\pi \mu
  \end{equation}
\end{proof}

Equation (\ref{PotEq-2}) is Poisson's equation and is the defining equation of the gravitational potential in Newtonian mechanics.  Moreover, equation (\ref{PotEq-1}) reduces to $V_{r} = M/r^{2}$ which yields the inverse square law for Newtonian gravity.  Then we can interpret $V$ as the analogue in our scenario of the Newtonian potential.  The interpretation of $M$ and $V$ via Propositions \ref{mH-Prop}, \ref{ED-Prop-1}, and \ref{PR-Prop-1} are what we mean by saying that the metric (\ref{metric-2}) is Newtonian compatible.

\subsection{Other Useful Properties of this Metric}

In this section, we produce two additional useful results which are readily obtainable with these coordinates and metric functions.  The first is the well-known result of the monotonicity of the Hawking mass in spherical symmetry, which is made particularly straightforward using this form of the metric.  It follows almost immediately from Corollary \ref{GCor-1}.

\begin{prop}\label{mH-Prop2}
  If the spacetime satisfies the dominant energy condition that $G(X,Y) \geq 0$ for all future-pointing causal (i.e. timelike or null) vector fields, $X$ and $Y$, then,  whenever $2M(t,r) \leq r$, the Hawking Mass, $M(t,r)$, is monotonically increasing in any non-timelike direction for which the radial coordinate increases.
\end{prop}

\begin{proof}
  The vector field $\nu_{t}$ is the future-pointing timelike unit vector field in the $t$ direction and the vector fields $\nu_{r} + \nu_{t}$ (future-pointing) and $\nu_{r} - \nu_{t}$ (past-pointing) are null vector fields in the null directions where the $r$ coordinate increases.  By Corollary \ref{GCor-1}, we have that
  \begin{align}
    \notag (\nu_{r} + \nu_{t})(M) &= M_{r}\pnth{1 - \frac{2M}{r}}^{1/2} + \frac{M_{t}}{ \e^{V}} \\
    \notag &= \frac{r^{2}}{2}\pnth{1 - \frac{2M}{r}}^{1/2}\pnth{\frac{2M_{r}}{r^{2}} + \frac{2M_{t}}{r^{2} \e^{V}}\pnth{1 - \frac{2M}{r}}^{-1/2}} \\
    \notag &= \frac{r^{2}}{2}\pnth{1 - \frac{2M}{r}}^{1/2}\pnth{G(\nu_{t},\nu_{t}) + G(\nu_{t},\nu_{r})} \\
    &= \frac{r^{2}}{2}\pnth{1 - \frac{2M}{r}}^{1/2} G(\nu_{t},\nu_{t}+\nu_{r})
  \end{align}
  The last factor here is positive by the dominant energy condition since both $\nu_{t}$ and $\nu_{t} + \nu_{r}$ are future-pointing causal vector fields.  Since $r^{2}\geq 0$ and $2M \leq r$ everywhere, it must be that $(\nu_{r} + \nu_{t})(M) \geq 0$ everywhere as well.  By the same corollary, we also have
  \begin{align}
    \notag (\nu_{r} - \nu_{t})(M) &= M_{r}\pnth{1 - \frac{2M}{r}}^{1/2} - \frac{M_{t}}{ \e^{V}} \\
    \notag &= \frac{r^{2}}{2}\pnth{1 - \frac{2M}{r}}^{1/2}\pnth{\frac{2M_{r}}{r^{2}} - \frac{2M_{t}}{r^{2} \e^{V}}\pnth{1 - \frac{2M}{r}}^{-1/2}} \\
    \notag &= \frac{r^{2}}{2}\pnth{1 - \frac{2M}{r}}^{1/2}\pnth{G(\nu_{t},\nu_{t}) - G(\nu_{t},\nu_{r})} \\
    &= \frac{r^{2}}{2}\pnth{1 - \frac{2M}{r}}^{1/2} G(\nu_{t},\nu_{t}-\nu_{r})
  \end{align}
  The last factor here is positive by the dominant energy condition since $\nu_{t} - \nu_{r}$ is also future-pointing causal.  Then, as before, it must be that $(\nu_{r} - \nu_{t})(M) \geq 0$ everywhere.  Since both $(\nu_{r} + \nu_{t})(M)$ and $(\nu_{r} - \nu_{t})(M)$ are both nonnegative everywhere, any positive linear combination of the two is also nonnegative, which is the desired result.
\end{proof}

This next property will be useful later when we consider the Einstein-Klein-Gordon system. 
We define the following functions, two of which have already been introduced in equations (\ref{ED-def}) and (\ref{PR-def}),
\begin{subequations} \label{SETFun-def}
\begin{align}
  \mu(t,r) &= T(\nu_{t}, \nu_{t}) & \rho(t,r) &= T(\nu_{t},\nu_{r}) \\
  P(t,r) &= T(\nu_{r},\nu_{r}) & Q(t,r) &= T(\nu_{\theta},\nu_{\theta}) = T(\nu_{\varphi},\nu_{\varphi}).
\end{align}
\end{subequations}
By the Einstein equation and Corollary \ref{GCor-1}, these functions account for all the nonzero components of $T$ in terms of the orthonormal frame $\{\nu_{t},\nu_{r},\nu_{\theta},\nu_{\varphi}\}$.  Then we have the following result which follows directly from the required property of all stress-energy tensors that $\divx{g} T = 0$, where $\divx{g} T$ denotes the divergence of the tensor $T$ with respect to the metric $g$.
We will leave the lengthy proof of this proposition to the appendix.

\begin{prop}\label{EinEq-prop1}
Suppose that the metric is spherically symmetric and of the form in equation (\ref{metric-2}) and that $T$ is a suitable spherically symmetric tensor of the form in equation (\ref{SETFun-def}).  Then solving the equation
\begin{equation}\label{divfreeT}
  \divx{g} T = 0
\end{equation}
along with solving the ODEs obtained from (\ref{EDEq-1}) and (\ref{PREq-1}), namely
  \begin{align}
    \label{EDEq-1-thm} M_{r} &= 4\pi r^{2}\mu \\
    \label{PREq-1-thm} V_{r} &= \pnth{1 - \frac{2M}{r}}^{-1}\pnth{\frac{M}{r^{2}} + 4\pi r P},
  \end{align}
  for each value of $t$ is equivalent to solving the entire Einstein equation.
\end{prop}

In this proposition, the requirement that $T$ is ``suitable'' refers to its having to come from some kind of physical notion or some set of mathematical conditions that yield a physical-like stress-energy tensor or, at least, a $T$ for which the Einstein equation is solvable.

\section{The Einstein-Klein-Gordon Equations}

A good example of the usefulness of this choice of metric is in the evolution of the Einstein-Klein-Gordon equations in spherical symmetry, which is the Einstein equation in the presence of a nontrivial scalar field.  In this paper, we will work specifically with a complex valued scalar field. Before we derive the spherically symmetric Einstein-Klein-Gordon equations in this metric, we should note that several other references including, but certainly not limited to, \cites{Hawley00, Seidel90, Seidel98, Bernal08, Gleiser89, MSBS, Matos09, Lee09, Sharma08, Lai07, Hawley03, Lee96} have written the Einstein-Klein-Gordon equations in spherical symmetry using either the metric presented here or another form of a general spherically symmetric metric.  Moreover, scalar fields, spherically symmetric or not, have been used as a sort of simple readily obtainable example to produce a nontrivial stress energy tensor, although it is rarely given physical significance \cite{Wald84}.  However, in the last couple decades or so, scalar fields have been considered as a candidate for dark matter by many authors, usually under the name boson stars, scalar field dark matter, Bose-Einstein condensates, or wave dark matter \cites{Bray10, Bray12, MSBS, Lee96, Lee09, Matos00, Matos09, Matos01}.  These scalar fields are sometimes real and sometimes complex depending on the discussion.  There are sometimes several scalar fields bound gravitationally as in \cite{MSBS} and sometimes just a single scalar field as in \cite{Bray10}.  This relatively new and actively studied application of scalar fields makes the  example presented here even more relevant and useful.  The treatment here most closely resembles the treatement given by Bray \cites{Bray10,Bray12}, although in his papers a real scalar field is usually considered and spherical symmetry is not necessarily assumed.  As such, we will reference his paper to make a few comparisons to the case of a real valued scalar field.

Consider a spherically symmetric spacetime $N$ with Lorentzian metric $g$ as described in the previous section.  Let $f:N\to\C$ be a spherically symmetric complex valued scalar field.  Then the Einstein-Klein Gordon equations in this case are
\begin{align}
  \label{EKG-1} G &= 8\pi \mu_{0}\pnth{\frac{df \otimes d\cf + d\cf \otimes df}{\Up^{2}} - \pnth{\frac{\abs{df}^{2}}{\Up^{2}} + \abs{f}^{2}}g} \\
  \label{EKG-2} \glap{g}f &= \Up^2 f
\end{align}
where $\mu_{0}$ is simply some constant that controls the scale of the system and is not to be confused with the energy density $\mu(t,0)$ at the central value.  Actually, the above equations are the Einstein-Klein-Gordon equations whether in spherical symmetry or not, but spherical symmetry is the general topic of this paper and so we require it here as well.  The value of $\mu_{0}$ is unimportant to the qualitative behavior of the solutions and can be absorbed entirely into $f$, if desired.  By contrast, the parameter $\Up$ is a fundamental constant to the equations upon which the qualitative behavior of the solutions depend.  The real valued version of these equations, that is, where $f$ is a real valued scalar field or equivalently where $\cf = f$, can be found in the previously mentioned paper by Bray \cite{Bray10} and it is readily seen that the equations here reduce to those in Bray's paper under the assumption $\cf = f$.  Moreover, equations (\ref{EKG-1}) and (\ref{EKG-2}) are the Euler-Lagrange equations of the action involving the scalar field given in Bray's paper if the scalar field is assumed to be complex valued instead of real valued (this follows the Lagrangian formulation of general relativity given in an appendix in Wald's book \cite{Wald84}).
Note also that equations (\ref{EinEq-1}) and (\ref{EKG-1}) imply that the stress-energy tensor corresponding to a complex scalar field is given by
\begin{equation} \label{SET-Def}
   T = \mu_{0}\pnth{\frac{df \otimes d\cf + d\cf \otimes df}{\Up^{2}} - \pnth{\frac{\abs{df}^{2}}{\Up^{2}} + \abs{f}^{2}}g}.
\end{equation}

With these equations, we will construct a system of PDEs which can be used to numerically evolve the scalar field $f$ and the metric from consistent initial data.  Since in the previous section we have already computed the components of the Einstein curvature tensor, in order to construct such a system of PDEs, we need to compute the components of the stress energy tensor and the Laplacian in the metric $g$.  Before we do, we will define the function, $p(t,r)$, by the equation
\begin{equation}\label{p-Def}
  p = f_{t}\e^{-V}\pnth{1 - \frac{2M}{r}}^{-1/2}.
\end{equation}
This is done to make the resulting system of PDEs first order in time and results in a more convenient choice than choosing only $p=f_{t}$.  We now have the following lemmas.
\begin{lem}\label{SET-Lem-1}
    For the stress energy tensor given in (\ref{SET-Def}), the following are true.
    \begin{enumerate}
      \item $\displaystyle T(\nu_{t},\nu_{t}) = \mu_{0}\pnth{\abs{f}^{2} + \pnth{1 - \frac{2M}{r}}\frac{\abs{f_{r}}^{2} + \abs{p}^{2}}{\Up^{2}}}$
      \item $\displaystyle T(\nu_{t},\nu_{r}) = \frac{2\mu_{0}}{\Up^2}\pnth{1 - \frac{2M}{r}}\Re(f_{r}\cp)$
      \item $\displaystyle T(\nu_{r},\nu_{r}) = \mu_{0}\pnth{-\abs{f}^{2} + \pnth{1 - \frac{2M}{r}}\frac{\abs{f_{r}}^{2} + \abs{p}^{2}}{\Up^{2}}}$
      \item $\displaystyle T(\nu_{\theta},\nu_{\theta}) = T(\nu_{\varphi},\nu_{\varphi}) = -\mu_{0}\pnth{\abs{f}^{2} + \pnth{1 - \frac{2M}{r}}\frac{\abs{f_{r}}^{2} - \abs{p}^{2}}{\Up^{2}}}$
    \end{enumerate}
\end{lem}

\noindent The proof of this lemma is left to the appendix.

\begin{lem}\label{EinEq-Lem-1}
    The components of the Einstein equation in the $\nu_{\eta}$ directions result in the following PDEs.
    \begin{align}
      \label{NCpde1} M_{r} &= 4\pi r^{2}\mu_{0}\pnth{\abs{f}^{2} + \pnth{1 - \frac{2M}{r}}\frac{\abs{f_{r}}^{2} + \abs{p}^{2}}{\Up^{2}}} \\
      \label{NCceq1} M_{t} &= \frac{8\pi r^{2}\mu_{0}\e^{V}}{\Up^{2}}\pnth{1 - \frac{2M}{r}}^{3/2}\Re(f_{r}\cp) \displaybreak[0]\\
      \label{NCpde2} V_{r} &= \pnth{1 - \frac{2M}{r}}^{-1}\pnth{\frac{M}{r^{2}} - 4\pi r\mu_{0}\pnth{\abs{f}^{2} - \pnth{1 - \frac{2M}{r}}\frac{\abs{f_{r}}^{2} + \abs{p}^{2}}{\Up^{2}}}} \\
      \notag V_{t}M_{t} &= -r\e^{2V}\Bigg[\pnth{V_{rr} + V_{r}^{2} + \frac{V_{r}}{r}}\pnth{1 - \frac{2M}{r}}^{2}  - \frac{3M_{t}^{2}}{r^{2}\e^{2V}}\pnth{1 - \frac{2M}{r}}^{-1}  - \frac{M_{tt}}{r\e^{2V}} \\
     \label{NCceq2} &\qquad + \pnth{\frac{M}{r} - M_{r}}\pnth{\frac{1}{r^{2}} + \frac{V_{r}}{r}}\pnth{1 - \frac{2M}{r}} 
     + 8\pi \mu_{0}\pnth{1 - \frac{2M}{r}}\pnth{\abs{f}^{2} + \pnth{1 - \frac{2M}{r}}\frac{\abs{f_{r}}^{2} - \abs{p}^{2}}{\Up^{2}}}\Bigg]
    \end{align}
\end{lem}

\begin{proof}  This follows directly from the Einstein equation (\ref{EinEq-1}), Corollary \ref{GCor-1}, and Lemma \ref{SET-Lem-1}.
Equations (\ref{NCpde1})-(\ref{NCceq2}) follow from the $(\nu_{t},\nu_{t})$, $(\nu_{t},\nu_{r})$, $(\nu_{r},\nu_{r})$, and $(\nu_{\theta},\nu_{\theta})$ (or $(\nu_{\varphi},\nu_{\varphi})$) components of the Einstein equation, respectively. \end{proof}

The Klein-Gordon equation and the fact that
\begin{equation}\label{glap-def}
  \glap{g}f = \frac{1}{\sqrt{\abs{g}}}\p_{\eta}\pnth{\sqrt{\abs{g}}g^{\eta\om}\p_{\om}f}.
\end{equation}
yield the following lemma.
  \begin{lem}\label{KG-Lem-1}
    The Klein-Gordon equation in this metric yields the following PDE.
    \begin{equation}\label{NCpde4}
      p_{t} = \e^{V}\pnth{-\Up^{2}f\pnth{1 - \frac{2M}{r}}^{-1/2} + \frac{2f_{r}}{r}\sqrt{1 - \frac{2M}{r}}} + \p_{r}\pnth{\e^{V}f_{r}\sqrt{1 - \frac{2M}{r}}}
    \end{equation}
  \end{lem}

\begin{proof}
  To compute the Laplacian, we will first need the quantity $\sqrt{\abs{g}}$.  We have
  \begin{equation}
    \sqrt{\abs{g}} = \sqrt{\e^{2V}\pnth{1 - \frac{2M}{r}}^{-1}r^{4}\sin^{2}\theta} = \e^{V}\pnth{1-\frac{2M}{r}}^{-1/2}r^{2}\sin\theta
  \end{equation}
  Then since $f = f(t,r)$ and $g$ is diagonal, we have by equation (\ref{glap-def})
  \begin{align}
    \notag \glap{g}f &= \frac{1}{\e^{V}r^{2}\sin\theta}\sqrt{1 - \frac{2M}{r}}\Bigg(\p_{t}\pnth{\e^{V}\pnth{1 - \frac{2M}{r}}^{-1/2}r^{2}\sin\theta\, g^{tt}f_{t}} \\
    \notag &\hspace{1.5 in} + \p_{r}\pnth{\e^{V}\pnth{1 - \frac{2M}{r}}^{-1/2}r^{2}\sin\theta\, g^{rr}f_{r}}\Bigg) \\
    \notag &= \frac{1}{\e^{V}r^{2}\sin\theta}\sqrt{1 - \frac{2M}{r}}\Bigg(\p_{t}\pnth{-p r^{2}\sin\theta} \\
    &\qquad + 2r\e^{V}\sqrt{1 - \frac{2M}{r}}\sin\theta\, f_{r} + r^{2}\sin\theta\, \p_{r}\pnth{\e^{V}f_{r}\sqrt{1 - \frac{2M}{r}}}\Bigg) 
  \end{align}
  Then the Klein-Gordon equation (\ref{EKG-2}) becomes
  \begin{align}
     \Up^{2}f &= -p_{t}\e^{-V}\sqrt{1 - \frac{2M}{r}} + \frac{2f_{r}}{r}\pnth{1 - \frac{2M}{r}} 
    + \e^{-V}\sqrt{1 - \frac{2M}{r}}\, \p_{r}\pnth{\e^{V}f_{r}\sqrt{1 - \frac{2M}{r}}}
  \end{align}
  Solving for $p_{t}$ yields the desired result.
\end{proof}

Then rewrite (\ref{p-Def}) as
\begin{equation} \label{NCpde3}
  f_{t} = p\e^{V}\sqrt{1 - \frac{2M}{r}}
\end{equation}
to yield an evolution equation for $f$.

The next lemma shows that $T$ given in (\ref{SET-Def}) is divergence free.
\begin{lem}\label{divTEKG-Lem}
   Equation (\ref{EKG-2}) implies that $\divx{g}T = 0$, where $T$ is given by (\ref{SET-Def}).
\end{lem}
\noindent This lemma's proof is also included in the appendix.  Now we can prove the following.
\begin{lem}\label{NCpde-Lem-1}
  If $f$, $p$, $M$, and $V$ satisfy (\ref{NCpde1}), (\ref{NCpde2}), (\ref{NCpde4}), and (\ref{NCpde3}), then they also satisfy (\ref{NCceq1}) and (\ref{NCceq2}).
\end{lem}

\begin{proof}
 By Lemma \ref{KG-Lem-1}, equations (\ref{NCpde4}) and (\ref{NCpde3}) are equivalent to the Klein-Gordon equation (\ref{EKG-2}) and hence, by Lemma \ref{divTEKG-Lem}, imply that $\divx{g}T = 0$, where $T$ is given by (\ref{SET-Def}).  By equation (\ref{SETFun-def}) and Lemma \ref{SET-Lem-1}, equations (\ref{NCpde1}) and (\ref{NCpde2}) are simply equations (\ref{EDEq-1-thm}) and (\ref{PREq-1-thm}) for the $T$ given by (\ref{SET-Def}).  Then by Proposition~\ref{EinEq-prop1}, $f$, $p$, $M$, and $V$ also satisfy the entire Einstein equation in this metric and for this stress energy tensor.  By Lemma \ref{EinEq-Lem-1}, this implies that $f$, $p$, $M$, and $V$ satisfy equations (\ref{NCceq1}) and (\ref{NCceq2}).
\end{proof}

Lemmas \ref{EinEq-Lem-1}, \ref{KG-Lem-1}, and \ref{NCpde-Lem-1} and equation (\ref{NCpde3}) prove the following proposition.
\begin{prop}
  Solving the following PDEs for $f(t,r)$, $p(t,r)$, $V(t,r)$, and $M(t,r)$ with consistent initial data is necessary and sufficient to solve the Einstein-Klein-Gordon equations (\ref{EKG-1}) and (\ref{EKG-2}) with a metric of the form (\ref{metric-2}) on the same initial data.
  \begin{subequations} \label{NCpde}
  \begin{align}
    \label{NCpde1a} M_{r} &= 4\pi r^{2}\mu_{0}\pnth{\abs{f}^{2} + \pnth{1 - \frac{2M}{r}}\frac{\abs{f_{r}}^{2} + \abs{p}^{2}}{\Up^{2}}} \\
    \label{NCpde2a} V_{r} &= \pnth{1 - \frac{2M}{r}}^{-1}\pnth{\frac{M}{r^{2}} - 4\pi r\mu_{0}\pnth{\abs{f}^{2} - \pnth{1 - \frac{2M}{r}}\frac{\abs{f_{r}}^{2} + \abs{p}^{2}}{\Up^{2}}}} \displaybreak[0]\\
    \label{NCpde3a} f_{t} &= p\e^{V}\sqrt{1 - \frac{2M}{r}} \\
    \label{NCpde4a} p_{t} &= \e^{V}\pnth{-\Up^{2}f\pnth{1 - \frac{2M}{r}}^{-1/2} + \frac{2f_{r}}{r}\sqrt{1 - \frac{2M}{r}}} + \p_{r}\pnth{\e^{V}f_{r}\sqrt{1 - \frac{2M}{r}}}
  \end{align}
  \end{subequations}
\end{prop}

Thus these equations are effectively the spherically symmetric Einstein-Klein-Gordon equations (specifically though only in the metric (\ref{metric-2})).

\section{Applications and Conclusion}

There are a great deal of applications to the collection of known results presented here.  This is especially true since, as stated before, spherical symmetry is usually the first testbed of any problem in general relativity, whether it be in regards to more classical topics such as the spacetime in the presence of a star, black holes, electromagnetism, and perfect fluids, or more recent ideas like dark matter or relativistic fluids.

One particular application which is well tailored to the information presented here has already been mentioned, namely, wave dark matter, that is, dark matter as a scalar field.  The entire previous section is essentially the basics of scalar fields in spherical symmetry and as such, it can be directly applied to that problem.  The author is currently working on this problem and of all the different coordinate systems and metric functions that we have used, this Newtonian compatible choice yields by far the simplest system of PDEs that we have come across that solve the spherically symmetric Einstein-Klein-Gordon equations.

An important step in the study of wave dark matter is to make comparisons to observations, especially in order to determine the value of the parameter $\Up$, the fundamental constant in the Einstein-Klein-Gordon equations.  Even showing the existence of some values of $\Up$ that are consistent with observed data would be an important contribution to the wave dark matter model.  One of the best places to start making these comparisons is with dwarf spheroidal galaxies, which are small, roughly spherically symmetric, and dark matter dominated.  The system (\ref{NCpde}) which solves the spherically symmetric Einstein-Klein-Gordon equations is particularly useful in this regard.  We hope to use these equations to generate solutions whose mass distribution is similar to those observered in these dwarf spheroidal galaxies.  The ensuing comparison should help us get an initial estimate on the value of  $\Up$.

This is just one of the many possible applications of spherical symmetry and the results surveyed in this paper.  With the importance of spherically symmetric spacetimes, there are sure to be many more applications in the future.  This paper is designed to be a stepping off point for these applications and can serve as a refresher or a general reference for the seasoned researcher or a crash course for a graduate student or someone using spherical symmetry in detail for the first time.

\appendix

\section{Proofs}\label{Proofs}

In this appendix, we present the proofs which were omitted in the above sections.  Every proof of a proposition or lemma from above will refer back to the number of the statement it is proving.  We also introduce a new lemma here which is useful for proving some of the above statements.  We will present all of the proofs here in the order with which their corresponding statements are given above.

\begin{proof}[Proof of Proposition \ref{mH-Prop}]
The Hawking mass of a metric sphere is defined as
\begin{equation}
  m_{H}(\Si_{t,r}) = \sqrt{\frac{\abs{\Si_{t,r}}}{16\pi}}\pnth{1 - \frac{1}{16\pi}\int_{\Si_{t,r}} g\pnth{\Hv,\Hv}\, dA}
\end{equation}
where $\abs{\Si_{t,r}}$ is the surface area of the metric sphere, $\Hv$ is the mean curvature vector of the sphere in the spacetime, and $dA$ is the volume element on the sphere \cites{Hawk68,Bray07}.  By the definition of our radial coordinate $r$, we have that $\abs{\Si_{t,r}}=4\pi r^{2}$, but it is also easily computed in this metric since
\begin{equation}\label{dA}
  dA = \sqrt{\abs{d\si^{2}}}\, d\theta \, d\varphi = \sqrt{r^{4} \sin^{2}\theta}\, d\theta \, d\varphi = r^{2}\sin \theta\, d\theta \, d\varphi,
\end{equation}
which yields that
\begin{equation}\label{Area}
  \abs{\Si_{t,r}} = \int_{\Si_{t,r}} dA = \int_{0}^{2\pi}\int_{0}^{\pi} r^{2}\sin \theta\, d\theta \, d\varphi = \int_{0}^{2\pi} 2r^{2}\, d\varphi = 4\pi r^{2}.
\end{equation}
To compute the mean curvature vector, $\Hv$, we have that
\begin{equation} \label{MCV-1}
  \Hv = \ga^{jk}\II(\p_{j},\p_{k})
\end{equation}
where $j,k \in \{\theta,\varphi\}$, $\ga$ is the metric on $\Si_{t,r}$ (note that we have reused the variable $\ga$ here to refer to the metric on the sphere and not to the metric on the $t=constant$ hypersurfaces as we did in the section describing the framework of numerical relativity), and $\II$ is the second fundamental form tensor, which sends a pair of vectors tangent to the sphere to a vector normal to the sphere.  It is defined as follows for all $X,Y \in T\Si_{t,r}$
\begin{equation}
  \II(X,Y) = \cov_{X}Y -\, ^{2}\cov_{X}Y
\end{equation}
where $\cov$ is the covariant derivative operator on $N$ and $^{2}\cov$ is the induced covariant derivative operator on $\Si_{t,r}$.  Since $\ga$ is diagonal, we only need to concern ourselves with the diagonal components of this tensor in order to compute $\Hv$.  To perform these computations, first note that the Christoffel symbols on the sphere in these coordinates have the following property (note that a superscript of 2 will always denote that that quantity corresponds to the sphere, also for the following computations, roman letters will denote subscripts in $\{\theta,\varphi\}$, while Greek letters will denote subscripts ranging over all four coordinates).  Since both the radial and time directions are normal to the sphere,
\begin{align}
  \Ga_{jk}^{\ \ \l} &= \frac{1}{2}g^{\l\eta}(g_{j\eta,k} + g_{k\eta,j} - g_{jk,\eta}) = \frac{1}{2}g^{\l m}(g_{jm,k} + g_{km,j} - g_{jk,m}) 
  = \frac{1}{2}\ga^{\l m}(g_{jm,k} + g_{km,j} - g_{jk,m}) = \,^{2}\Ga_{jk}^{\ \ \l}.
\end{align}
Thus we need not distinguish between the Christoffel symbols for the metric on the sphere and those on the entire manifold.  We have then that
\begin{equation}
  \II(\p_{j},\p_{j}) = \cov_{j}\p_{j} -\, ^{2}\cov_{j}\p_{j} = \Ga_{jj}^{\ \ \ \eta}\, \p_{\eta} -\, ^{2}\Ga_{jj}^{\ \ \ k}\, \p_{k} = \Ga_{jj}^{\ \ \ t}\, \p_{t} + \Ga_{jj}^{\ \ \ r}\, \p_{r}.
\end{equation}
Then we need to compute the above Christoffel symbols.  We have that
\begin{subequations} \label{Chistoffel-1}
\begin{align}
  \Ga_{\theta\theta}^{\ \ \ t} &= \frac{1}{2}g^{t\eta}(g_{\theta \eta,\theta } + g_{\theta \eta,\theta } - g_{\theta \theta ,\eta}) = \frac{1}{2}g^{tt}(- g_{\theta \theta ,t}) = -\frac{1}{2 \e^{2V}}(-\p_{t}(r^{2})) = 0 \\
  \Ga_{\theta\theta}^{\ \ \ r} &= \frac{1}{2}g^{r\eta}(g_{\theta \eta,\theta } + g_{\theta \eta,\theta } - g_{\theta \theta ,\eta}) = \frac{1}{2}g^{rr}(- g_{\theta \theta ,r}) 
  = \frac{1}{2}\pnth{1-\frac{2M}{r}}(-\p_{r}(r^{2})) = -r\pnth{1 - \frac{2M}{r}} \\
  \Ga_{\varphi\varphi}^{\ \ \ t} &= \frac{1}{2}g^{t\eta}(g_{\varphi \eta,\varphi } + g_{\varphi \eta,\varphi } - g_{\varphi \varphi ,\eta}) = \frac{1}{2}g^{tt}(- g_{\varphi \varphi ,t}) = -\frac{1}{2 \e^{2V}}(-\p_{t}(r^{2}\sin^{2}\theta)) = 0 \\
  \Ga_{\varphi\varphi}^{\ \ \ r} &= \frac{1}{2}g^{r\eta}(g_{\varphi \eta,\varphi } + g_{\varphi \eta,\varphi } - g_{\varphi \varphi ,\eta}) = \frac{1}{2}g^{rr}(- g_{\varphi \varphi ,r}) 
  = \frac{1}{2}\pnth{1-\frac{2M}{r}}(-\p_{r}(r^{2}\sin^{2}\theta)) = -r\pnth{1 - \frac{2M}{r}}\sin^{2}\theta.
\end{align}
\end{subequations}
Then we have that
\begin{align}
  \II(\p_{\theta},\p_{\theta}) &= \Ga_{\theta\theta}^{\ \ \ t}\, \p_{t} + \Ga_{\theta\theta}^{\ \ \ r}\, \p_{r} = -r\pnth{1 - \frac{2M}{r}}\, \p_{r} \\
  \II(\p_{\varphi},\p_{\varphi}) &= \Ga_{\varphi\varphi}^{\ \ \ t}\, \p_{t} + \Ga_{\varphi\varphi}^{\ \ \ r}\, \p_{r} = -r\pnth{1 - \frac{2M}{r}}\sin^{2}\theta\, \p_{r}.
\end{align}
This makes (\ref{MCV-1}) become
\begin{equation}
  \label{MCV-2} \Hv = \ga^{\theta\theta}\II(\p_{\theta},\p_{\theta}) + \ga^{\varphi\varphi}\II(\p_{\varphi},\p_{\varphi})
  = -\frac{2}{r}\pnth{1 - \frac{2M}{r}}\, \p_{r}.
\end{equation}
Finally, we can compute the Hawking mass as follows.
\begin{align}
  \notag m_{H}(\Si_{t,r}) &= \sqrt{\frac{\abs{\Si_{t,r}}}{16\pi}}\pnth{1 - \frac{1}{16\pi}\int_{\Si_{t,r}} g\pnth{\Hv,\Hv}\, dA} \\
  \notag &= \sqrt{\frac{4\pi r^{2}}{16\pi}}\pnth{1 - \frac{1}{16\pi}\int_{\Si_{t,r}} g\pnth{-\frac{2}{r}\pnth{1 - \frac{2M}{r}}\, \p_{r},\, -\frac{2}{r}\pnth{1 - \frac{2M}{r}}\, \p_{r}}\, dA} \displaybreak[0]\\
  \notag &= \frac{r}{2}\pnth{1 - \frac{1}{16\pi}\int_{0}^{2\pi}\int_{0}^{\pi} \frac{4}{r^{2}}\pnth{1 - \frac{2M}{r}}^{2}g(\p_{r},\p_{r}) r^{2}\sin\theta \, d\theta\, d\varphi} \displaybreak[0]\\
  \notag &= \frac{r}{2}\pnth{1 - \frac{1}{16\pi}\int_{0}^{2\pi}\int_{0}^{\pi} 4\pnth{1 - \frac{2M}{r}} \sin\theta \, d\theta\, d\varphi} \\
  \label{mH-1} m_{H}(\Si_{t,r}) &= M
\end{align}
\end{proof}

To prove Lemma \ref{RicLem-2}, we will utilize the following additional lemma.

\begin{lem} \label{RicLem-1}
  For all $\eta \in \{t,r,\theta,\varphi\}$, let $\nu_{\eta}$ be as defined by equation (\ref{norvec}).  Then the following are true.
    \begin{enumerate}
      \item\label{RicLem-1a} For all $\eta \in \{t,r,\theta,\varphi\}$, $\Ric(\nu_{\eta},\nu_{\eta}) = g(\nu_{\eta},\nu_{\eta})g^{\eta\eta}\Ric_{\eta\eta}$.
      \item\label{RicLem-1b} $\displaystyle R = \sum_{\eta} g(\nu_{\eta},\nu_{\eta})\Ric(\nu_{\eta},\nu_{\eta})$
    \end{enumerate}
\end{lem}

\begin{proof}
Since $\{\nu_{\eta}\}$ is a frame field, we have that
\begin{equation}
  \Ric(\nu_{t},\nu_{t}) =  \e^{-2V}\Ric(\p_{t},\p_{t}) = -g^{tt}\Ric_{tt} = g(\nu_{t},\nu_{t})g^{tt}\Ric_{tt}.
\end{equation}
By similar arguments, we have that for all $\eta \in \{t,r,\theta,\varphi\}$,
\begin{equation}
  \Ric(\nu_{\eta},\nu_{\eta}) = g(\nu_{\eta},\nu_{\eta})g^{\eta\eta}\Ric_{\eta\eta}
\end{equation}
which proves \ref{RicLem-1a}.

To prove \ref{RicLem-1b}, we use the first result and find that since $g$ is diagonal and $g(\nu_{\eta},\nu_{\eta})=\pm 1$, we have that
\begin{equation}
  R = \sum_{\eta\om} g^{\eta\om}\Ric_{\eta\om} = \sum_{\eta} g^{\eta\eta}\Ric_{\eta\eta} = \sum_{\eta} \frac{\Ric(\nu_{\eta},\nu_{\eta})}{g(\nu_{\eta},\nu_{\eta})} = \sum_{\eta} g(\nu_{\eta},\nu_{\eta})\Ric(\nu_{\eta},\nu_{\eta}).
\end{equation}
Fact \ref{RicLem-1b} also follows from a well understood more general property of frame fields and traces of tensors.
\end{proof}

\noindent With this, we prove Lemma \ref{RicLem-2}.

\begin{proof}[Proof of Lemma \ref{RicLem-2}]
In this proof, all indices range over the coordinates $\{t,r,\theta,\varphi\}$.  By Lemma \ref{RicLem-1} and the well known formula for the components of the Riemann curvature tensor, $\R$, in terms of the Christoffel symbols
\begin{equation}
  \R_{jk\l}^{\ \ \ \ m} = \Ga_{j\l\ \ ,k}^{\ \ m} - \Ga_{k\l\ \ ,j}^{\ \ m} + \Ga_{j\l}^{\ \ s}\Ga_{ks}^{\ \ m} - \Ga_{k\l}^{\ \ s}\Ga_{js}^{\ \ m}.
\end{equation}
Since $\Ric_{jk} = \R_{j\l k}^{\ \ \ \ \l}$, we have that
\begin{subequations} \label{RicEq-1}
  \begin{align}
    \Ric(\nu_{t},\nu_{t}) &= g(\nu_{t},\nu_{t})g^{tt}\Ric_{tt} 
    =  \e^{-2V}\pnth{\Ga_{tt\ ,k}^{\ \ k} - \Ga_{kt\ ,t}^{\ \ k} + \Ga_{tt}^{\ \ s}\Ga_{ks}^{\ \ k} - \Ga_{kt}^{\ \ s}\Ga_{ts}^{\ \ k}} \displaybreak[0]\\
    \Ric(\nu_{t},\nu_{r}) &=  \e^{-V}\sqrt{1 - \frac{2M}{r}}\, \Ric_{tr} 
    =  \e^{-V}\sqrt{1 - \frac{2M}{r}} \pnth{\Ga_{tr\ ,k}^{\ \ k} - \Ga_{kr\ ,t}^{\ \ k} + \Ga_{tr}^{\ \ s}\Ga_{ks}^{\ \ k} - \Ga_{kr}^{\ \ s}\Ga_{ts}^{\ \ k}} \displaybreak[0]\\
    \Ric(\nu_{r},\nu_{r}) &= g(\nu_{r},\nu_{r})g^{rr}\Ric_{rr} = \pnth{1 - \frac{2M}{r}}\Ric_{rr} 
    = \pnth{1 - \frac{2M}{r}} \pnth{\Ga_{rr\ ,k}^{\ \ k} - \Ga_{kr\ ,r}^{\ \ k} + \Ga_{rr}^{\ \ s}\Ga_{ks}^{\ \ k} - \Ga_{kr}^{\ \ s}\Ga_{rs}^{\ \ k}} \displaybreak[0]\\
    \Ric(\nu_{\theta},\nu_{\theta}) &= g(\nu_{\theta},\nu_{\theta})g^{\theta\theta}\Ric_{\theta\theta} = \frac{1}{r^{2}}\Ric_{\theta\theta} 
    = \frac{1}{r^{2}} \pnth{\Ga_{\theta \theta \ ,k}^{\ \ k} - \Ga_{k\theta \ ,\theta }^{\ \ k} + \Ga_{\theta \theta }^{\ \ s}\Ga_{ks}^{\ \ k} - \Ga_{k\theta }^{\ \ s}\Ga_{\theta s}^{\ \ k}} \displaybreak[0]\\
    \Ric(\nu_{\varphi},\nu_{\varphi}) &= g(\nu_{\varphi},\nu_{\varphi})g^{\varphi\varphi}\Ric_{\varphi\varphi} = \frac{1}{r^{2}\sin^{2} \theta}\Ric_{\varphi\varphi} 
    = \frac{1}{r^{2}\sin^{2}\theta} \pnth{\Ga_{\varphi \varphi \ ,k}^{\ \ \ k} - \Ga_{k\varphi \ ,\varphi }^{\ \ \ k} + \Ga_{\varphi \varphi }^{\ \ \ s}\Ga_{ks}^{\ \ k} - \Ga_{k\varphi }^{\ \ \ s}\Ga_{\varphi s}^{\ \ \ k}}.
  \end{align}
\end{subequations}
Then to compute these, we need the Christoffel symbols.
Recall that the formula for the Christoffel symbols in terms of the metric components is the following
\begin{equation} \label{CS-Def}
  \Ga_{jk}^{\ \ \l} = \frac{1}{2}g^{\l m} (g_{jm,k} + g_{km,j} - g_{jk,m})
\end{equation}
Using (\ref{CS-Def}) and the fact that $g$ is diagonal, we have that
\begin{subequations} \label{CS-1}
\begin{align}
  \Ga_{jk}^{\ \ t} &= \frac{1}{2}g^{t m}(g_{jm,k} + g_{km,j} - g_{jk,m}) = \frac{1}{2}g^{tt}(g_{jt,k} + g_{kt,j} - g_{jk,t}) 
  = -\frac{1}{2 \e^{2V}}(g_{jt,k} + g_{kt,j} - g_{jk,t}) \displaybreak[0]\\
  \Ga_{jk}^{\ \ r} &= \frac{1}{2}g^{r m}(g_{jm,k} + g_{km,j} - g_{jk,m}) = \frac{1}{2}g^{rr}(g_{jr,k} + g_{kr,j} - g_{jk,r}) 
  = \frac{1}{2}\pnth{1 - \frac{2M}{r}} (g_{jr,k} + g_{kr,j} - g_{jk,r}) \displaybreak[0]\\
  \Ga_{jk}^{\ \ \theta}&=  \frac{1}{2}g^{\theta m}(g_{jm,k} + g_{km,j} - g_{jk,m}) = \frac{1}{2}g^{\theta\theta}(g_{j\theta,k} + g_{k\theta,j} - g_{jk,\theta}) 
  = \frac{1}{2r^{2}} (g_{j\theta,k} + g_{k\theta,j} - g_{jk,\theta}) \displaybreak[0]\\
  \notag \Ga_{jk}^{\ \ \varphi} &= \frac{1}{2}g^{\varphi m}(g_{jm,k} + g_{km,j} - g_{jk,m}) = \frac{1}{2}g^{\varphi\varphi}(g_{j\varphi,k} + g_{k\varphi,j} - g_{jk,\varphi}) \\
  &= \frac{1}{2r^{2}\sin^{2}\theta} (g_{j\varphi,k} + g_{k\varphi,j} - g_{jk,\varphi}) = \frac{1}{2r^{2}\sin^{2}\theta} (g_{j\varphi,k} + g_{k\varphi,j}).
\end{align}
\end{subequations}
The last line is due to the fact that none of the metric components depend on $\varphi$.  To help compute these quantities, it would be useful to note the following.
\begin{subequations} \label{MDer}
\begin{align}
  \notag g_{tt,t} &= \p_{t}(- \e^{2V}) = -2V_{t} \e^{2V} & g_{tt,\theta} &= \p_{\theta}(- \e^{2V}) = 0 \\
  g_{tt,r} &= \p_{r}(- \e^{2V}) = -2V_{r} \e^{2V} & g_{tt,\varphi} &= \p_{\varphi}(- \e^{2V}) = 0 \displaybreak[0] \\
  \notag g_{rr,t} &= \p_{t}\pnth{1 - \frac{2M}{r}}^{-1} = \pnth{1 - \frac{2M}{r}}^{-2}\pnth{\frac{2M_{t}}{r}} &  g_{rr,\theta} &= \p_{\theta}\pnth{1 - \frac{2M}{r}}^{-1} = 0 \\
  g_{rr,r} &= \p_{r}\pnth{1 - \frac{2M}{r}}^{-1} = \pnth{1 - \frac{2M}{r}}^{-2}\pnth{\frac{2M_{r}}{r} - \frac{2M}{r^{2}}} & g_{rr,\varphi} &= \p_{\varphi}\pnth{1 - \frac{2M}{r}}^{-1} = 0 \displaybreak[0] \\
  \notag g_{\theta\theta,t} &= \p_{t}(r^{2}) = 0 & g_{\theta\theta,\theta} &= \p_{\theta}(r^{2}) = 0 \\
  g_{\theta\theta,r} &= \p_{r}(r^{2}) = 2r & g_{\theta\theta,\varphi} &= \p_{\varphi}(r^{2}) = 0 \displaybreak[0] \\
  \notag g_{\varphi\varphi,t} &= \p_{t}(r^{2}\sin^{2}\theta) = 0 & g_{\varphi\varphi,\theta} &= \p_{\theta}(r^{2}\sin^{2}\theta) = 2r^{2}\sin\theta \cos\theta \\
  g_{\varphi\varphi,r} &= \p_{r}(r^{2}\sin^{2}\theta) = 2r\sin^{2}\theta & g_{\varphi\varphi,\varphi} &= \p_{\varphi}{r^{2}\sin^{2}\theta} = 0.
\end{align}
\end{subequations}
All the other metric components are 0.  Next, if
\begin{equation}
  \Phi = \pnth{1 - \frac{2M}{r}},
\end{equation}
then, using (\ref{CS-1}) and (\ref{MDer}), it is a series of straightforward algebra computations to obtain the following values of the Christoffel symbols.  For brevity, we only include the nonzero Christoffel symbols.
\begin{subequations} \label{CS-2}
  \begin{align}
    \Ga_{tt}^{\ \ t} &= V_{t} & \Ga_{tr}^{\ \ t} &= \Ga_{rt}^{\ \ t} = V_{r} & \Ga_{rr}^{\ \ t} &= \frac{M_{t}}{r\e^{2V}}\Phi^{-2} \\
    \notag \Ga_{tt}^{\ \ r} &= V_{r}\e^{2V}\Phi & \Ga_{tr}^{\ \ r} &= \Ga_{rt}^{\ \ r} = \frac{M_{t}}{r}\Phi^{-1} & \Ga_{rr}^{\ \ r} &= \Phi^{-1}\pnth{\dfrac{M_{r}}{r} - \dfrac{M}{r^{2}}} \\
    \Ga_{\theta\theta}^{\ \ \, r} &= -r\Phi & \Ga_{\varphi\varphi}^{\ \ \ r} &= -r \Phi\sin^{2}\theta & & \\
    \Ga_{r\theta}^{\ \ \, \theta} &= \Ga_{\theta r}^{\ \ \, \theta} = \frac{1}{r} & \Ga_{\varphi \varphi}^{ \ \ \ \theta} &= -\sin \theta \cos \theta & & \\
    \Ga_{r\varphi}^{\ \ \ \varphi} & = \Ga_{\varphi r}^{\ \ \ \varphi} = \frac{1}{r} & \Ga_{\theta \varphi}^{\ \ \ \varphi} &= \Ga_{\varphi \theta}^{\ \ \ \varphi} =  \dfrac{\cos\theta}{\sin\theta} & &
  \end{align}
\end{subequations}
Using (\ref{CS-2}), we can now continue the computations started in (\ref{RicEq-1}).
\begin{align}
  \notag \Ric(\nu_{t},\nu_{t}) &=  \e^{-2V}\pnth{\Ga_{tt\ ,k}^{\ \ k} - \Ga_{kt\ ,t}^{\ \ k} + \Ga_{tt}^{\ \ s}\Ga_{ks}^{\ \ k} - \Ga_{kt}^{\ \ s}\Ga_{ts}^{\ \ k}} \displaybreak[0]\\
  \notag &=  \e^{-2V}\Bigg\{\p_{t}V_{t} + \p_{r}\pnth{V_{r} \e^{2V}\pnth{1 - \frac{2M}{r}}} - \p_{t}V_{t} - \p_{t}\pnth{\frac{M_{t}}{r}\pnth{1 - \frac{2M}{r}}^{-1}} 
  + V_{t}\pnth{V_{t} + \frac{M_{t}}{r}\pnth{1 - \frac{2M}{r}}^{-1}} \\
  \notag &\qquad + V_{r} \e^{2V}\pnth{1 - \frac{2M}{r}}\pnth{V_{r} + \pnth{1 - \frac{2M}{r}}^{-1}\pnth{\frac{M_{r}}{r} - \frac{M}{r^{2}}} + \frac{2}{r}} \\
  \notag &\qquad - V_{t}^{2} - 2 V_{r}^{2} \e^{2V}\pnth{1 - \frac{2M}{r}} - \frac{M_{t}^{2}}{r^{2}}\pnth{1 - \frac{2M}{r}}^{-2} \Bigg\} \displaybreak[0]\\
  \notag &=  \e^{-2V}\Bigg\{V_{rr} \e^{2V}\pnth{1 - \frac{2M}{r}} + 2V_{r}^{2} \e^{2V}\pnth{1 - \frac{2M}{r}} + V_{r} \e^{2V}\pnth{\frac{2M}{r^{2}} - \frac{2M_{r}}{r}} \\
  \notag &\qquad - \frac{M_{tt}}{r}\pnth{1 - \frac{2M}{r}}^{-1} - \frac{2M_{t}^{2}}{r^{2}}\pnth{1 - \frac{2M}{r}}^{-2} + \frac{V_{t}M_{t}}{r}\pnth{1 - \frac{2M}{r}}^{-1} \\
  \notag &\qquad + V_{r}^{2} \e^{2V}\pnth{1 - \frac{2M}{r}} + V_{r} \e^{2V}\pnth{\frac{M_{r}}{r} - \frac{M}{r^{2}}} + \frac{2V_{r} \e^{2V}}{r}\pnth{1 - \frac{2M}{r}} \\
  \notag &\qquad - 2 V_{r}^{2} \e^{2V}\pnth{1 - \frac{2M}{r}} - \frac{M_{t}^{2}}{r^{2}}\pnth{1 - \frac{2M}{r}}^{-2} \Bigg\} \displaybreak[0]\\
  \Ric(\nu_{t},\nu_{t}) &= \pnth{V_{rr} + V_{r}^{2} + \frac{2V_{r}}{r}}\pnth{1-\frac{2M}{r}} + \frac{V_{r}}{r}\pnth{\frac{M}{r} - M_{r}} 
  + \frac{V_{t}M_{t} - M_{tt}}{r \e^{2V}}\pnth{1-\frac{2M}{r}}^{-1} - \frac{3M_{t}^{2}}{r^{2} \e^{2V}}\pnth{1-\frac{2M}{r}}^{-2}
\end{align}
\begin{align}
  \notag \Ric(\nu_{t},\nu_{r}) &=  \e^{-V}\sqrt{1 - \frac{2M}{r}} \pnth{\Ga_{tr\ ,k}^{\ \ k} - \Ga_{kr\ ,t}^{\ \ k} + \Ga_{tr}^{\ \ s}\Ga_{ks}^{\ \ k} - \Ga_{kr}^{\ \ s}\Ga_{ts}^{\ \ k}} 
  \\
  \notag &=  \e^{-V}\sqrt{1 - \frac{2M}{r}} \Bigg\{\p_{t}(V_{r}) + \p_{r}\pnth{\frac{M_{t}}{r}\pnth{1 - \frac{2M}{r}}^{-1}} - \p_{t}(V_{r}) \\
  \notag &\qquad - \p_{t}\pnth{\pnth{1 - \frac{2M}{r}}^{-1}\pnth{\frac{M_{r}}{r} - \frac{M}{r^{2}}}} + V_{r}\pnth{V_{t} + \frac{M_{t}}{r}\pnth{1 - \frac{2M}{r}}^{-1}} \\
  \notag &\qquad - 2\p_{t}\pnth{\frac{1}{r}} + \frac{M_{t}}{r}\pnth{1 - \frac{2M}{r}}^{-1}\pnth{V_{r} + \pnth{1 - \frac{2M}{r}}^{-1}\pnth{\frac{M_{r}}{r} - \frac{M}{r^{2}}} + \frac{2}{r}} \\
  \notag &\qquad - V_{r}V_{t} - \frac{2M_{t}V_{r}}{r}\pnth{1 - \frac{2M}{r}}^{-1} - \frac{M_{t}}{r}\pnth{1 - \frac{2M}{r}}^{-2}\pnth{\frac{M_{r}}{r} - \frac{M}{r^{2}}} \Bigg\} \displaybreak[0]\\
  \notag &=  \e^{-V}\sqrt{1 - \frac{2M}{r}} \Bigg\{\p_{r}\p_{t}\pnth{-\frac{1}{2}\ln\pnth{1 - \frac{2M}{r}}} - \p_{t}\p_{r}\pnth{-\frac{1}{2}\ln\pnth{1 - \frac{2M}{r}}} 
  + \frac{2M_{t}}{r^{2}}\pnth{1 - \frac{2M}{r}}^{-1} \Bigg\} \\
  \Ric(\nu_{t},\nu_{r}) &= \frac{2M_{t}}{r^{2} \e^{V}}\pnth{1 - \frac{2M}{r}}^{-1/2}
\end{align}

\begin{align}
  \notag \Ric(\nu_{r},\nu_{r}) &= \pnth{1 - \frac{2M}{r}} \pnth{\Ga_{rr\ ,k}^{\ \ k} - \Ga_{kr\ ,r}^{\ \ k} + \Ga_{rr}^{\ \ s}\Ga_{ks}^{\ \ k} - \Ga_{kr}^{\ \ s}\Ga_{rs}^{\ \ k}} \displaybreak[0] \\
  \notag &= \pnth{1 - \frac{2M}{r}} \Bigg\{\p_{t}\pnth{\frac{M_{t}}{r \e^{2V}}\pnth{1 - \frac{2M}{r}}^{-2}} + \p_{r}\pnth{\pnth{1 - \frac{2M}{r}}^{-1}\pnth{\frac{M_{r}}{r} - \frac{M}{r^{2}}}} \\
  \notag &\qquad - \p_{r}V_{r} -\p_{r}\pnth{\pnth{1 - \frac{2M}{r}}^{-1}\pnth{\frac{M_{r}}{r} - \frac{M}{r^{2}}}} \\
  \notag &\qquad - 2\p_{r}\pnth{\frac{1}{r}} + \frac{V_{t}M_{t}}{r \e^{2V}}\pnth{1 - \frac{2M}{r}}^{-2} + \frac{M_{t}^{2}}{r^{2} \e^{2V}}\pnth{1 - \frac{2M}{r}}^{-3} \\
  \notag &\qquad + \pnth{V_{r} + \frac{2}{r}}\pnth{1 - \frac{2M}{r}}^{-1}\pnth{\frac{M_{r}}{r} - \frac{M}{r^{2}}} + \pnth{1 - \frac{2M}{r}}^{-2}\pnth{\frac{M_{r}}{r} - \frac{M}{r^{2}}}^{2} \\
  \notag &\qquad - V_{r}^{2} - \frac{2M_{t}^{2}}{r^{2} \e^{2V}}\pnth{1 - \frac{2M}{r}}^{-3} - \frac{2}{r^{2}} - \pnth{1 - \frac{2M}{r}}^{-2}\pnth{\frac{M_{r}}{r} - \frac{M}{r^{2}}}^{2} \Bigg\} \displaybreak[0] \\
  \notag &= \pnth{1 - \frac{2M}{r}} \Bigg\{\frac{M_{tt}}{r \e^{2V}}\pnth{1 - \frac{2M}{r}}^{-2} - \frac{2V_{t}M_{t}}{r \e^{2V}}\pnth{1 - \frac{2M}{r}}^{-2} \\
  \notag &\qquad + \frac{4M_{t}^{2}}{r^{2} \e^{2V}}\pnth{1 - \frac{2M}{r}}^{-3} - V_{rr} + \frac{2}{r^{2}} + \frac{V_{t}M_{t}}{r \e^{2V}}\pnth{1 - \frac{2M}{r}}^{-2} \\
  \notag &\qquad - \frac{M_{t}^{2}}{r^{2} \e^{2V}}\pnth{1 - \frac{2M}{r}}^{-3} - V_{r}^{2} - \frac{2}{r^{2}} 
  + \pnth{V_{r} + \frac{2}{r}}\pnth{1 - \frac{2M}{r}}^{-1}\pnth{\frac{M_{r}}{r} - \frac{M}{r^{2}}} \Bigg\} \displaybreak[0] \\
  \Ric(\nu_{r},\nu_{r}) &= -(V_{rr} + V_{r}^2)\pnth{1 - \frac{2M}{r}} - \pnth{\frac{2}{r^{2}} + \frac{V_{r}}{r}}\pnth{\frac{M}{r} - M_{r}} 
  - \frac{V_{t}M_{t} - M_{tt}}{r \e^{2V}}\pnth{1-\frac{2M}{r}}^{-1} + \frac{3M_{t}^{2}}{r^2 \e^{2V}}\pnth{1-\frac{2M}{r}}^{-2}
\end{align}
\begin{align}
  \notag \Ric(\nu_{\theta},\nu_{\theta}) &= \frac{1}{r^{2}} \pnth{\Ga_{\theta \theta \ ,k}^{\ \ k} - \Ga_{k\theta \ ,\theta }^{\ \ k} + \Ga_{\theta \theta }^{\ \ s}\Ga_{ks}^{\ \ k} - \Ga_{k\theta }^{\ \ s}\Ga_{\theta s}^{\ \ k}} \displaybreak[0]\\
  \notag &= \frac{1}{r^{2}}\Bigg\{\p_{r}\pnth{-r\pnth{1 - \frac{2M}{r}}} - \p_{\theta}\pnth{\frac{\cos \theta}{\sin\theta}} - r\pnth{1 - \frac{2M}{r}}\pnth{V_{r} + \frac{2}{r}} \\
  \notag &\qquad - r\pnth{\frac{M_{r}}{r} - \frac{M}{r^{2}}} + 2\pnth{1 - \frac{2M}{r}} - \frac{\cos^{2}\theta}{\sin^{2}\theta}\Bigg\} \displaybreak[0]\\
  \notag &= \frac{1}{r^{2}}\Bigg\{-\pnth{1 - \frac{2M}{r}} - r\pnth{\frac{2M}{r^{2}} - \frac{2M_{r}}{r}} + \frac{1}{\sin^{2}\theta} - r\pnth{1 - \frac{2M}{r}}\pnth{V_{r} + \frac{2}{r}} \\
  \notag &\qquad + r\pnth{\frac{M}{r^{2}} - \frac{M_{r}}{r}} + 2\pnth{1 - \frac{2M}{r}} - \frac{\cos^{2}\theta}{\sin^{2}\theta}\Bigg\} \displaybreak[0] \\
  \notag &= \frac{1}{r^{2}}\Bigg\{-\frac{\sin^{2}\theta}{\sin^{2}\theta} + \frac{2M}{r} - \frac{M}{r} + M_{r} + \frac{1}{\sin^{2}\theta} - rV_{r}\pnth{1 - \frac{2M}{r}} \\
  \notag &\qquad - 2\pnth{1 - \frac{2M}{r}} + 2 \pnth{1 - \frac{2M}{r}} - \frac{\cos^{2}\theta}{\sin^{2}\theta} \Bigg\} \\
  \Ric(\nu_{\theta},\nu_{\theta}) &= \frac{1}{r^{2}}\pnth{\frac{M}{r} + M_{r}} - \frac{V_{r}}{r}\pnth{1 - \frac{2M}{r}}
\end{align}
\begin{align}
  \notag \Ric(\nu_{\varphi},\nu_{\varphi}) &= \frac{1}{r^{2}\sin^{2}\theta} \pnth{\Ga_{\varphi \varphi \ ,k}^{\ \ \ k} - \Ga_{k\varphi \ ,\varphi }^{\ \ \ k} + \Ga_{\varphi \varphi }^{\ \ \ s}\Ga_{ks}^{\ \ k} - \Ga_{k\varphi }^{\ \ \ s}\Ga_{\varphi s}^{\ \ \ k}} \displaybreak[0] \\
  \notag &= \frac{1}{r^{2}\sin^{2}\theta} \Bigg\{\p_{r}\pnth{-r\sin^{2}\theta\pnth{1 - \frac{2M}{r}}} + \p_{\theta}\pnth{-\sin\theta\cos\theta} \\
  \notag &\qquad - r\sin^{2}\theta \pnth{1 - \frac{2M}{r}}\pnth{V_{r} + \frac{2}{r}} - r\sin^{2}\theta\pnth{\frac{M_{r}}{r} - \frac{M}{r^{2}}} - \cos^{2}\theta \\
  \notag &\qquad + 2\sin^{2}\theta \pnth{1 - \frac{2M}{r}} + 2\cos^{2}\theta \Bigg\} \displaybreak[0] \\
  \notag &= \frac{1}{r^{2}\sin^{2}\theta} \Bigg\{-\sin^{2}\theta\pnth{1 - \frac{2M}{r}} - r\sin^{2}\theta\pnth{\frac{2M}{r^{2}} - \frac{2M_{r}}{r}} - \cos^{2}\theta + \sin^{2}\theta \\
  \notag &\qquad - r\sin^{2}\theta \pnth{1 - \frac{2M}{r}}\pnth{V_{r} + \frac{2}{r}} + r\sin^{2}\theta\pnth{\frac{M}{r^{2}} - \frac{M_{r}}{r}} - \cos^{2}\theta \\
  \notag &\qquad + 2\sin^{2}\theta \pnth{1 - \frac{2M}{r}} + 2\cos^{2}\theta \Bigg\} \displaybreak[0] \\
  \notag &= \frac{1}{r^{2}}\pnth{\frac{2M}{r} - \frac{M}{r} + M_{r} - rV_{r}\pnth{1 - \frac{2M}{r}} - 2\pnth{1 - \frac{2M}{r}} + 2\pnth{1 - \frac{2M}{r}}} \\
  \label{Ricphph} \Ric(\nu_{\varphi},\nu_{\varphi}) &= \frac{1}{r^{2}}\pnth{\frac{M}{r} + M_{r}} - \frac{V_{r}}{r}\pnth{1 - \frac{2M}{r}} = \Ric(\nu_{\theta},\nu_{\theta}).
\end{align}
These are the only nonzero components of the Ricci curvature tensor since
\begin{align}
  \notag \Ric(\nu_{t},\nu_{\theta}) &= \frac{1}{r\e^{V}}\Ric_{t\theta} = \frac{1}{r\e^{V}}\R_{tk\theta}^{\ \ \ \, k} \\
  \notag &= \frac{1}{r\e^{V}}\pnth{\Ga_{t\theta\ ,k}^{\ \ \, k} - \Ga_{k\theta\ ,t}^{\ \ \, k} + \Ga_{t\theta}^{\ \ \, s}\Ga_{ks}^{\ \ k} - \Ga_{k\theta}^{\ \ \, s}\Ga_{ts}^{\ \ k}} \displaybreak[0]\\
  \notag &= \frac{1}{r\e^{V}}\pnth{- \Ga_{r\theta\ ,t}^{\ \ \, r} - \Ga_{\theta\theta\ ,t}^{\ \ \ \theta} - \Ga_{\varphi\theta\ ,t}^{\ \ \ \varphi} - \Ga_{k\theta}^{\ \ \, r}\Ga_{tr}^{\ \ k} - \Ga_{k\theta}^{\ \ \, \theta}\Ga_{t\theta}^{\ \ \, k} - \Ga_{k\theta}^{\ \ \, \varphi}\Ga_{t\varphi}^{\ \ \, k}} \displaybreak[0]\\
  \notag &= \frac{1}{r\e^{V}}\pnth{- \Ga_{\theta\theta}^{\ \ \ r}\Ga_{tr}^{\ \ \theta} - \Ga_{r\theta}^{\ \ \, \theta}\Ga_{t\theta}^{\ \ \, r} - \Ga_{\varphi\theta}^{\ \ \ \varphi}\Ga_{t\varphi}^{\ \ \, \varphi}} \\
  &= 0
\end{align}
\begin{align}
  \notag \Ric(\nu_{t},\nu_{\varphi}) &= \frac{1}{r\sin \theta \e^{V}}\Ric_{t\varphi} = \frac{1}{r\sin \theta\e^{V}}\R_{tk\varphi}^{\ \ \ \, k} \\
  \notag &= \frac{1}{r\sin \theta\e^{V}}\pnth{\Ga_{t\varphi\ ,k}^{\ \ \, k} - \Ga_{k\varphi\ ,t}^{\ \ \, k} + \Ga_{t\varphi}^{\ \ \, s}\Ga_{ks}^{\ \ k} - \Ga_{k\varphi}^{\ \ \, s}\Ga_{ts}^{\ \ k}} \displaybreak[0]\\
  \notag &= \frac{1}{r\sin \theta\e^{V}}\pnth{- \Ga_{r\varphi\ ,t}^{\ \ \, r} - \Ga_{\theta\varphi\ ,t}^{\ \ \ \theta} - \Ga_{\varphi\varphi\ ,t}^{\ \ \ \varphi} - \Ga_{k\varphi}^{\ \ \, r}\Ga_{tr}^{\ \ k} - \Ga_{k\varphi}^{\ \ \, \theta}\Ga_{t\theta}^{\ \ \, k} - \Ga_{k\varphi}^{\ \ \, \varphi}\Ga_{t\varphi}^{\ \ \, k}} \displaybreak[0]\\
  \notag &= \frac{1}{r\sin \theta\e^{V}}\pnth{- \Ga_{\varphi\varphi}^{\ \ \ r}\Ga_{tr}^{\ \ \varphi} - \Ga_{\varphi\varphi}^{\ \ \ \theta}\Ga_{t\theta}^{\ \ \, \varphi} - \Ga_{r\varphi}^{\ \ \, \varphi}\Ga_{t\varphi}^{\ \ \, r} - \Ga_{\theta\varphi}^{\ \ \ \varphi}\Ga_{t\varphi}^{\ \ \, \theta}} \\
  &= 0
\end{align}
\begin{align}
  \notag \Ric(\nu_{r},\nu_{\theta}) &= \frac{1}{r}\sqrt{1 - \frac{2M}{r}}\Ric_{r\theta} = \frac{1}{r}\sqrt{1 - \frac{2M}{r}}\R_{rk\theta}^{\ \ \ \, k} \\
  \notag &= \frac{1}{r}\sqrt{1 - \frac{2M}{r}}\pnth{\Ga_{r\theta\ ,k}^{\ \ \, k} - \Ga_{k\theta\ ,r}^{\ \ \, k} + \Ga_{r\theta}^{\ \ \, s}\Ga_{ks}^{\ \ k} - \Ga_{k\theta}^{\ \ \, s}\Ga_{rs}^{\ \ k}} \displaybreak[0] \\
  \notag &= \frac{1}{r}\sqrt{1 - \frac{2M}{r}}(\Ga_{r\theta\ ,\theta}^{\ \ \, \theta} - \Ga_{r\theta\ ,r}^{\ \ \, r} - \Ga_{\theta\theta\ ,r}^{\ \ \ \theta} - \Ga_{\varphi\theta\ ,r}^{\ \ \ \varphi} + \Ga_{r\theta}^{\ \ \, \theta}\Ga_{r\theta}^{\ \ \, r} \\
  \notag &\qquad + \Ga_{r\theta}^{\ \ \, \theta}\Ga_{\theta\theta}^{\ \ \ \theta} + \Ga_{r\theta}^{\ \ \, \theta}\Ga_{\varphi\theta}^{\ \ \varphi} - \Ga_{r\theta}^{\ \ \, s}\Ga_{rs}^{\ \ r} - \Ga_{\theta\theta}^{\ \ \ s}\Ga_{rs}^{\ \ \theta} - \Ga_{\varphi\theta}^{\ \ \ s}\Ga_{rs}^{\ \ \varphi}) \displaybreak[0] \\
  \notag &= \frac{1}{r}\sqrt{1 - \frac{2M}{r}}\pnth{\frac{\cos \theta}{r\sin \theta} - \Ga_{r\theta}^{\ \ \, \theta}\Ga_{r\theta}^{\ \ \, r} - \Ga_{\theta\theta}^{\ \ \ r}\Ga_{rr}^{\ \ \theta} - \Ga_{\varphi\theta}^{\ \ \ \varphi}\Ga_{r\varphi}^{\ \ \, \varphi}} \displaybreak[0] \\
  \notag &= \frac{1}{r}\sqrt{1 - \frac{2M}{r}}\pnth{\frac{\cos \theta}{r\sin \theta} - \frac{\cos \theta}{r\sin\theta}} \\
  &= 0
\end{align}
\begin{align}
  \notag \Ric(\nu_{r},\nu_{\varphi}) &= \frac{1}{r\sin \theta}\sqrt{1 - \frac{2M}{r}}\Ric_{r\varphi} = \frac{1}{r\sin \theta}\sqrt{1 - \frac{2M}{r}}\R_{rk\varphi}^{\ \ \ \, k} \\
  \notag &= \frac{1}{r\sin\theta}\sqrt{1 - \frac{2M}{r}}\pnth{\Ga_{r\varphi\ ,k}^{\ \ \, k} - \Ga_{k\varphi\ ,r}^{\ \ \, k} + \Ga_{r\varphi}^{\ \ \, s}\Ga_{ks}^{\ \ k} - \Ga_{k\varphi}^{\ \ \, s}\Ga_{rs}^{\ \ k}} \displaybreak[0] \\
  \notag &= \frac{1}{r\sin\theta}\sqrt{1 - \frac{2M}{r}}(\Ga_{r\varphi\ ,\varphi}^{\ \ \, \varphi} - \Ga_{r\varphi\ ,r}^{\ \ \, r} - \Ga_{\theta\varphi\ ,r}^{\ \ \ \theta} - \Ga_{\varphi\varphi\ ,r}^{\ \ \ \varphi} + \Ga_{r\varphi}^{\ \ \, \varphi}\Ga_{r\varphi}^{\ \ r} \\
  \notag &\qquad + \Ga_{r\varphi}^{\ \ \, \varphi}\Ga_{\theta\varphi}^{\ \ \ \theta} + \Ga_{r\varphi}^{\ \ \, \varphi}\Ga_{\varphi\varphi}^{\ \ \ \varphi} - \Ga_{r\varphi}^{\ \ \, s}\Ga_{rs}^{\ \ r} - \Ga_{\theta\varphi}^{\ \ \ s}\Ga_{rs}^{\ \ \theta} - \Ga_{\varphi\varphi}^{\ \ \ s}\Ga_{rs}^{\ \ \varphi}) \displaybreak[0] \\
  \notag &= \frac{1}{r\sin\theta}\sqrt{1 - \frac{2M}{r}}\pnth{- \Ga_{r\varphi}^{\ \ \, \varphi}\Ga_{r\varphi}^{\ \ \, r} - \Ga_{\theta\varphi}^{\ \ \ \varphi}\Ga_{r\varphi}^{\ \ \, \theta} - \Ga_{\varphi\varphi}^{\ \ \ r}\Ga_{rr}^{\ \ \varphi} - \Ga_{\varphi\varphi}^{\ \ \ \theta}\Ga_{r\theta}^{\ \ \, \varphi}} \\
  &= 0
\end{align}
\begin{align}
  \notag \Ric(\nu_{\theta},\nu_{\varphi}) &= \frac{1}{r^{2}\sin\theta}\Ric_{\theta\varphi} = \frac{1}{r^{2}\sin \theta}\R_{\theta k \varphi}^{\ \ \ \ k} \\
  \notag &= \frac{1}{r^{2}\sin\theta}\pnth{\Ga_{\theta\varphi\ ,k}^{\ \ \ k} - \Ga_{k\varphi\ ,\theta}^{\ \ \, k} + \Ga_{\theta\varphi}^{\ \ \ s}\Ga_{ks}^{\ \ k} - \Ga_{k\varphi}^{\ \ \, s}\Ga_{\theta s}^{\ \ \, k}} \displaybreak[0] \\
  \notag &= \frac{1}{r^{2}\sin\theta}(\Ga_{\theta\varphi\ ,\varphi}^{\ \ \ \varphi} - \Ga_{r\varphi\ ,\theta}^{\ \ \, r} - \Ga_{\theta\varphi\ ,\theta}^{\ \ \ \theta} - \Ga_{\varphi\varphi\ ,\theta}^{\ \ \ \varphi} + \Ga_{\theta\varphi}^{\ \ \ \varphi}\Ga_{r\varphi}^{\ \ \, r} \\
  \notag &\qquad + \Ga_{\theta\varphi}^{\ \ \ \varphi}\Ga_{\theta\varphi}^{\ \ \ \theta} + \Ga_{\theta\varphi}^{\ \ \ \varphi}\Ga_{\varphi\varphi}^{\ \ \ \varphi} - \Ga_{k\varphi}^{\ \ \, r}\Ga_{\theta r}^{\ \ \, k} - \Ga_{k\varphi}^{\ \ \, \theta}\Ga_{\theta \theta}^{\ \ \ k} - \Ga_{k\varphi}^{\ \ \, \varphi}\Ga_{\theta \varphi}^{\ \ \ k}) \displaybreak[0] \\
  \notag &= \frac{1}{r^{2}\sin\theta}\pnth{- \Ga_{\theta\varphi}^{\ \ \ r}\Ga_{\theta r}^{\ \ \, \theta} - \Ga_{r\varphi}^{\ \ \, \theta}\Ga_{\theta \theta}^{\ \ \ r} - \Ga_{\varphi\varphi}^{\ \ \ \varphi}\Ga_{\theta \varphi}^{\ \ \ \varphi}} \\
  &= 0
\end{align}
The other components are 0 by the symmetry of the Ricci curvature tensor.  To get the last statement, we note that by equation (\ref{Ricphph}) and Lemma \ref{RicLem-1}-\ref{RicLem-1b}, we have that
\begin{equation}
  R = \sum_{\eta} g(\nu_{\eta},\nu_{\eta})\Ric(\nu_{\eta},\nu_{\eta}) = -\Ric(\nu_{t},\nu_{t}) + \Ric(\nu_{r},\nu_{r}) + 2\Ric(\nu_{\theta},\nu_{\theta}).
\end{equation}
From here, it is a matter of algebra and the statements above to get the desired equation for $R$.
\end{proof}

Next, we prove Proposition \ref{EinEq-prop1}.

\begin{proof}[Proof of Proposition \ref{EinEq-prop1}]

In solving equations (\ref{divfreeT}), (\ref{EDEq-1-thm}), and (\ref{PREq-1-thm}), we already solve the equations resulting from two of the components of the Einstein equation, namely (\ref{EDEq-1-thm}) and (\ref{PREq-1-thm}), which we reprint here for convenience.
  \begin{align}
    \label{EDEq-1-rp} M_{r} &= 4\pi r^{2}\mu \\
    \label{PREq-1-rp} V_{r} &= \pnth{1 - \frac{2M}{r}}^{-1}\pnth{\frac{M}{r^{2}} + 4\pi r P}
  \end{align}
In order to show that solving these two equations coupled with $\divx{g} T = 0$ is sufficient to solve the entire Einstein equation, we must show that if these equations hold, so do the other components of the Einstein equation.  To do so, we first need to write down the other components of the Einstein equation.  Recall that the above equations come from the $\nu_{t},\nu_{t}$ and $\nu_{r},\nu_{r}$ components of the Einstein equation.  Thus we only have the $\nu_{t},\nu_{r}$ and $\nu_{\theta},\nu_{\theta}$ components left to compute (the $\nu_{\varphi},\nu_{\varphi}$ component is identical to $\nu_{\theta},\nu_{\theta}$ component).  Then, by Corollary \ref{GCor-1} and equation (\ref{SETFun-def}), we have the following.
\begin{align}
  \notag G(\nu_{t},\nu_{r}) &= 8\pi T(\nu_{t},\nu_{r}) \\
  \label{rhoEq-1} M_{t} &= 4\pi r^{2} \e^{V}\rho\sqrt{1 - \frac{2M}{r}}
\end{align}
and
\begin{align}
  \notag G(\nu_{\theta},\nu_{\theta}) &= 8\pi T(\nu_{\theta},\nu_{\theta}) \\
  \notag 8\pi Q &= \pnth{V_{rr} + V_{r}^{2} + \frac{V_{r}}{r}}\pnth{1-\frac{2M}{r}} + \pnth{\frac{M}{r} - M_{r}}\pnth{\frac{1}{r^{2}} + \frac{V_{r}}{r}} \\
  \label{QEq-1} &\qquad + \frac{V_{t}M_{t} - M_{tt}}{r \e^{2V}}\pnth{1 - \frac{2M}{r}}^{-1} - \frac{3M_{t}^{2}}{r^{2} \e^{2V}}\pnth{1-\frac{2M}{r}}^{-2}
\end{align}

Next 
we compute what $\divx{g} T = 0$
means in terms of the functions $\mu$, $P$, $\rho$, and $Q$.  In what follows, the summing indices run through the set $\{t,r,\theta,\varphi\}$ and we will use the Einstein summation convention wherever it applies, while specifically denoting any other summations of a different form.  First, we define
\begin{equation}
  \vep_{k} = g(\nu_{k},\nu_{k})
\end{equation}
 for all $k \in \{t,r,\theta,\varphi\}$.  Then we can write the divergence of $T$ as follows.
\begin{align}
  \notag \divx{g} T &= \sum_{k} \vep_{k}(\cov_{\nu_{k}}T)(\nu_{k},\p_{j})\, dx^{j} \\
  \label{divT1} &= -(\cov_{\nu_{t}}T)(\nu_{t},\p_{j})\, dx^{j} + (\cov_{\nu_{r}}T)(\nu_{r},\p_{j})\, dx^{j} 
  + (\cov_{\nu_{\theta}}T)(\nu_{\theta},\p_{j})\, dx^{j} + (\cov_{\nu_{\varphi}}T)(\nu_{\varphi},\p_{j})\, dx^{j}
\end{align}
We will simplify each of these four terms individually.  By equations (\ref{CS-2}), (\ref{norvec}), and (\ref{SETFun-def}), we have that
\begin{align}
  \notag (\cov_{\nu_{t}}T)(\nu_{t},\p_{k})\, dx^{k} &=  \e^{-2V}(\cov_{t}T)(\p_{t},\p_{k})\, dx^{k} \\
  \notag &=  \e^{-2V}\Big[\p_{t}(T(\p_{t},\p_{k})) - \Ga_{tt}^{\ \ m}T(\p_{m},\p_{k}) - \Ga_{tk}^{\ \ m}T(\p_{t},\p_{m})\Big]\, dx^{k} \displaybreak[0]\\
  \notag &=  \e^{-2V}\Big[\p_{t}(T(\p_{t},\p_{t}))\, dt + \p_{t}(T(\p_{t},\p_{r}))\, dr 
  - 2\Ga_{tt}^{\ \ t}T(\p_{t},\p_{t})\, dt - 2\Ga_{tt}^{\ \ r}T(\p_{t},\p_{r})\, dt \\
  \notag &\qquad - (\Ga_{tt}^{\ \ t} + \Ga_{tr}^{\ \ r})T(\p_{t},\p_{r})\, dr - \Ga_{tr}^{\ \ t}T(\p_{t},\p_{t})\, dr - \Ga_{tt}^{\ \ r}T(\p_{r},\p_{r})\, dr\Big] \displaybreak[0]\\
  \notag &=  \e^{-2V}\Bigg[\pnth{2V_{t} \e^{2V}\mu +  \e^{2V}\mu_{t} - 2V_{t} \e^{2V}\mu - 2\rho V_{r} \e^{3V}\sqrt{1 - \frac{2M}{r}}}\, dt \\
  \notag &\qquad + \Bigg(V_{t} \e^{V}\rho\pnth{1 - \frac{2M}{r}}^{-1/2} +  \e^{V}\rho_{t}\pnth{1 - \frac{2M}{r}}^{-1/2} 
  + \frac{ \e^{V}\rho M_{t}}{r}\pnth{1 - \frac{2M}{r}}^{-3/2}  \\
  \notag &\qquad - \pnth{V_{t} + \frac{M_{t}}{r}\pnth{1 - \frac{2M}{r}}^{-1}} \e^{V}\rho\pnth{1 - \frac{2M}{r}}^{-1/2} 
  - V_{r} \e^{2V}\mu - V_{r} \e^{2V}P\Bigg)\, dr \Bigg] \displaybreak[0] \\
  \label{C_tT} (\cov_{\nu_{t}}T)(\nu_{t},\p_{k})\, dx^{k} &= \pnth{\mu_{t} - 2\rho V_{r} \e^{V}\sqrt{1 - \frac{2M}{r}}}\, dt 
  + \pnth{\rho_{t} \e^{-V}\pnth{1 - \frac{2M}{r}}^{-1/2} - V_{r}(\mu + P)}\, dr
\end{align}
Next, we have
\begin{align}
  \notag (\cov_{\nu_{r}}T)(\nu_{r},\p_{k})\, dx^{k} &= \pnth{1 - \frac{2M}{r}}(\cov_{r}T)(\p_{r},\p_{k})\, dx^{k} \\
  \notag &= \pnth{1 - \frac{2M}{r}}\Big[\p_{r}(T(\p_{r},\p_{k})) - \Ga_{rr}^{\ \ m}T(\p_{m},\p_{k}) - \Ga_{rk}^{\ \ m}T(\p_{r},\p_{m})\Big]\, dx^{k} \displaybreak[0] \\
  \notag &= \pnth{1 - \frac{2M}{r}}\Big[\p_{r}(T(\p_{r},\p_{t}))\, dt + \p_{r}(T(\p_{r},\p_{r}))\, dr \\
  \notag &\qquad - \pnth{\Ga_{rr}^{\ \ r} + \Ga_{rt}^{\ \ t}}T(\p_{t},\p_{r})\, dt - 2\Ga_{rr}^{\ \ t}T(\p_{t},\p_{r})\, dr \\
  \notag &\qquad - 2\Ga_{rr}^{\ \ r}T(\p_{r},\p_{r})\, dr - \Ga_{rr}^{\ \ t}T(\p_{t},\p_{t})\, dt - \Ga_{rt}^{\ \ r}T(\p_{r},\p_{r})\, dt \Big] \displaybreak[0] \\
  \notag &= \pnth{1 - \frac{2M}{r}}\Bigg[\Bigg(\rho  \e^{V}\pnth{1 - \frac{2M}{r}}^{-3/2}\pnth{\frac{M_{r}}{r} - \frac{M}{r^{2}}} \\
  \notag &\qquad + \rho V_{r} \e^{V}\pnth{1 - \frac{2M}{r}}^{-1/2} + \rho_{r} \e^{V}\pnth{1 - \frac{2M}{r}}^{-1/2} \\
  \notag &\qquad - \frac{\mu M_{t}}{r}\pnth{1 - \frac{2M}{r}}^{-2} - \frac{P M_{t}}{r}\pnth{1 - \frac{2M}{r}}^{-2} \\
  \notag &\qquad -\rho  \e^{V}\pnth{1 - \frac{2M}{r}}^{-1/2}\brkt{\pnth{1 - \frac{2M}{r}}^{-1}\pnth{\frac{M_{r}}{r} - \frac{M}{r^{2}}} + V_{r}}\Bigg)\, dt \\
  \notag &\qquad + \Bigg(P_{r}\pnth{1 - \frac{2M}{r}}^{-1} + 2P\pnth{1 - \frac{2M}{r}}^{-2}\pnth{\frac{M_{r}}{r} - \frac{M}{r^{2}}} \\
  \notag &\qquad - \frac{2\rho M_{t}}{r  \e^{V}}\pnth{1 - \frac{2M}{r}}^{-5/2} - 2P\pnth{1 - \frac{2M}{r}}^{-2}\pnth{\frac{M_{r}}{r} - \frac{M}{r^{2}}} \Bigg)\, dr \Bigg] \displaybreak[0] \\
  \label{C_rT} (\cov_{\nu_{r}}T)(\nu_{r},\p_{k})\, dx^{k} &= \pnth{\rho_{r} \e^{V}\sqrt{1 - \frac{2M}{r}} - \frac{M_{t}}{r}\pnth{1 - \frac{2M}{r}}^{-1}(\mu + P)}\, dt 
  + \pnth{P_{r} - \frac{2\rho M_{t}}{r  \e^{V}}\pnth{1 - \frac{2M}{r}}^{-3/2}}\, dr
\end{align}
Thirdly,
\begin{align}
  \notag (\cov_{\nu_{\theta}}T)(\nu_{\theta},\p_{k})\, dx^{k} &= \frac{1}{r^{2}}(\cov_{\theta}T)(\p_{\theta},\p_{k})\, dx^{k} \\
  \notag &= \frac{1}{r^{2}} \Big[\p_{\theta}(T(\p_{\theta},\p_{k})) - \Ga_{\theta\theta}^{\ \ m}T(\p_{m},\p_{k}) - \Ga_{\theta k}^{\ \ m}T(\p_{\theta},\p_{m})\Big]\, dx^{k} \displaybreak[0]\\
  \notag &= \frac{1}{r^{2}}\Big[-\Ga_{\theta\theta}^{\ \ r}T(\p_{r},\p_{t})\, dt - \Ga_{\theta\theta}^{\ \ r}T(\p_{r},\p_{r})\, dr - \Ga_{\theta r}^{\ \ \theta}T(\p_{\theta},\p_{\theta})\, dr \Big] \displaybreak[0]\\
  \notag &= \frac{1}{r^{2}}\brkt{r\rho  \e^{V}\sqrt{1 - \frac{2M}{r}}\, dt + r P \, dr - r Q\, dr} \\
  \label{C_thT} (\cov_{\nu_{\theta}}T)(\nu_{\theta},\p_{k})\, dx^{k} &= \frac{\rho  \e^{V}}{r}\sqrt{1 - \frac{2M}{r}}\, dt + \frac{P - Q}{r}\, dr
\end{align}
Lastly, we have that
\begin{align}
  \notag (\cov_{\nu_{\varphi}}T)(\nu_{\varphi},\p_{k})\, dx^{k} &= \frac{1}{r^{2}\sin^{2}\theta}(\cov_{\varphi}T)(\p_{\varphi},\p_{k})\, dx^{k} \\
  \notag &= \frac{1}{r^{2}\sin^{2}\theta}\Big[ \p_{\varphi}(T(\p_{\varphi},\p_{k})) - \Ga_{\varphi\varphi}^{\ \ \, m}T(\p_{m},\p_{k}) - \Ga_{\varphi k}^{\ \ m}T(\p_{\varphi},\p_{m}) \Big]\, dx^{k} \displaybreak[0]\\
  \notag &= \frac{1}{r^{2}\sin^{2}\theta} \Big[ - \Ga_{\varphi\varphi}^{\ \ \, r}T(\p_{r},\p_{t})\, dt - \Ga_{\varphi\varphi}^{\ \ \, r}T(\p_{r},\p_{r})\, dr - \Ga_{\varphi\varphi}^{\ \ \, \theta}T(\p_{\theta},\p_{\theta})\, d\theta \\
  \notag &\qquad - \Ga_{\varphi r}^{\ \ \, \varphi}T(\p_{\varphi},\p_{\varphi})\, dr - \Ga_{\varphi\theta}^{\ \ \, \varphi}T(\p_{\varphi},\p_{\varphi})\, d\theta \Big] \displaybreak[0]\\
  \notag &= \frac{1}{r^{2}\sin^{2}\theta} \Big[ r\rho  \e^{V} \sin^{2}\theta \sqrt{1 - \frac{2M}{r}}\, dt + r \sin^{2}\theta P\, dr \\
  \notag &\qquad + r^{2}Q\sin\theta \cos\theta \, d\theta - r Q \sin^{2}\theta \, dr - r^{2}Q \cos \theta \sin \theta\, d\theta \Big] \displaybreak[0]\\
  \notag &= \frac{\rho  \e^{V}}{r}\sqrt{1 - \frac{2M}{r}}\, dt + \frac{P - Q}{r}\, dr \\
  \label{C_phT} (\cov_{\nu_{\varphi}}T)(\nu_{\varphi},\p_{k})\, dx^{k} &= (\cov_{\nu_{\theta}}T)(\nu_{\theta},\p_{k})\, dx^{k}
\end{align}
Then equation (\ref{divT1}) becomes
\begin{align}
  \notag \divx{g} T &= -(\cov_{\nu_{t}}T)(\nu_{t},\p_{j})\, dx^{j} + (\cov_{\nu_{r}}T)(\nu_{r},\p_{j})\, dx^{j} 
  + (\cov_{\nu_{\theta}}T)(\nu_{\theta},\p_{j})\, dx^{j} + (\cov_{\nu_{\varphi}}T)(\nu_{\varphi},\p_{j})\, dx^{j} \\
  \notag &= \Bigg(-\mu_{t} + 2\rho  \e^{V}\sqrt{1 - \frac{2M}{r}}\pnth{V_{r} + \frac{1}{r}} + \rho_{r} \e^{V}\sqrt{1 - \frac{2M}{r}} - \frac{M_{t}}{r}\pnth{1 - \frac{2M}{r}}^{-1}(\mu + P)\Bigg)\, dt  \\
  \label{divT2} &\qquad + \Bigg(-\rho_{t} \e^{-V}\pnth{1 - \frac{2M}{r}}^{-1/2} 
  + V_{r}(\mu + P) + P_{r} - \frac{2\rho M_{t}}{r  \e^{V}}\pnth{1 - \frac{2M}{r}}^{-3/2} + \frac{2(P - Q)}{r}\Bigg)\, dr
\end{align}
Then $\divx{g} T = 0$ yields the following two equations
\begin{subequations} \label{divT3}
  \begin{align}
    \label{divT3a} \mu_{t} &= 2\rho  \e^{V}\sqrt{1 - \frac{2M}{r}}\pnth{V_{r} + \frac{1}{r}} + \rho_{r} \e^{V}\sqrt{1 - \frac{2M}{r}} - \frac{M_{t}}{r}\pnth{1 - \frac{2M}{r}}^{-1}(\mu + P) \displaybreak[0]\\
    \label{divT3b} \rho_{t} &=  \e^{V}\sqrt{1 - \frac{2M}{r}}\pnth{V_{r}(\mu + P) + P_{r} + \frac{2(P - Q)}{r}} - \frac{2\rho M_{t}}{r}\pnth{1 - \frac{2M}{r}}^{-1}
  \end{align}
\end{subequations}

We are now ready to show that solving $\divx{g} T = 0$ 
along with solving equations (\ref{EDEq-1-rp}) and (\ref{PREq-1-rp}) on each $t=constant$ slice 
also solves the remaining two unique nonzero components of the Einstein equation, (\ref{rhoEq-1}) and (\ref{QEq-1}).  We will show first that (\ref{EDEq-1-rp}), (\ref{PREq-1-rp}), and (\ref{divT3}) implies that (\ref{rhoEq-1}) holds as well.  To do this, we first need to show that (\ref{EDEq-1-rp}) and (\ref{rhoEq-1}) are compatible, that is, given these equations, $M_{rt} = M_{tr}$.  Differentiating (\ref{EDEq-1-rp}) with respect to $t$ yields,
\begin{align}
  \notag \p_{t}M_{r} &= \p_{t}\pnth{4\pi r^{2}\mu} \\
  \label{EDEq-1_t} M_{rt} &= 4\pi r^{2}\mu_{t}
\end{align}
while differentiating (\ref{rhoEq-1}) with respect to $r$ yields,
\begin{align}
  \notag \p_{r}M_{t} &= \p_{r}\pnth{4\pi r^{2} \e^{V}\rho\sqrt{1 - \frac{2M}{r}}} \\
  \notag M_{tr} &= 4\pi\Bigg[ 2r  \e^{V} \rho \sqrt{1 - \frac{2M}{r}} + r^{2}V_{r} \e^{V}\rho\sqrt{1 - \frac{2M}{r}} 
  + r^{2} \e^{V}\rho_{r}\sqrt{1 - \frac{2M}{r}} - r^{2} \e^{V}\rho\pnth{1 - \frac{2M}{r}}^{-1/2}\pnth{\frac{M_{r}}{r} - \frac{M}{r^{2}}} \Bigg] \\
  \notag &= 4\pi r^{2}\Bigg[\mu_{t} + \frac{M_{t}}{r}\pnth{1 - \frac{2M}{r}}^{-1}(\mu + P) - 4\pi r \rho  \e^{V}\pnth{1 - \frac{2M}{r}}^{-1/2}(\mu + P) \Bigg] \\
  \label{rhoEq-1_r} &= 4\pi r^{2}\mu_{t}
\end{align}
where in the next to last line we made substitutions using equations (\ref{EDEq-1-rp}), (\ref{PREq-1-rp}), and (\ref{divT3a}) and the last line used (\ref{rhoEq-1}).  Thus if $\divx{g}T=0$, (\ref{EDEq-1-rp}) and (\ref{rhoEq-1}) are compatible, that is, $M_{tr} = M_{rt}$ everywhere the equations are defined, namely, wherever $r \neq 0$.  However, by using \LHopital's rule on the equations above coupled with the fact from equation (\ref{Mder}) that at $r = 0$, $M = M_{r} = M_{rr} = 0$ for all $t$, the equation $M_{tr} = M_{rt}$ still holds at the central value $r=0$.  This implies that there exists a function $M(t,r)$ which satisfies both (\ref{EDEq-1-rp}) and (\ref{PREq-1-rp}) everywhere, which of course is the metric function we seek.

As per our hypothesis, if for all values of $t$, we have solved (\ref{EDEq-1-rp}) with the initial condition $M = 0$ at $r=0$, we will obtain some function $M^{\ast}(t,r)$ which satisfies (\ref{EDEq-1-rp}) everywhere.  Then we have that 
  \begin{equation}
    M_{r}^{\ast} = M_{r}
  \end{equation}
everywhere, where $M$ is the function that satisfies both (\ref{EDEq-1-rp}) and (\ref{PREq-1-rp}).  This implies that 
  \begin{equation}
    M^{\ast}(t,r) = M(t,r) + f(t)
  \end{equation}
for some smooth function $f(t)$.  However, we also have that $M^{\ast}(t,0) = M(t,0) = 0$ for all $t$, which implies that $f(t) \equiv 0$.  Hence $M^{\ast}(t,r) = M(t,r)$ and the function we obtain by integrating (\ref{EDEq-1-rp}) with compatible initial conditions necessarily also satisfies (\ref{rhoEq-1}).

Finally, we use equations (\ref{EDEq-1-rp}), (\ref{PREq-1-rp}), (\ref{divT3}), and (\ref{rhoEq-1}) (since the first three imply equation (\ref{rhoEq-1})) to show that (\ref{QEq-1}) is automatically satisfied.  In order to do this, we will have to compute $V_{rr}$ and $M_{tt}$.  They are
\begin{align}
  \notag \p_{t} M_{t} &= \p_{t}\pnth{4\pi r^{2} \e^{V}\rho\sqrt{1 - \frac{2M}{r}}} \\
  \notag M_{tt} &= 4\pi r^{2}\brkt{V_{t} \e^{V}\rho \sqrt{1 - \frac{2M}{r}} +  \e^{V}\rho_{t}\sqrt{1 - \frac{2M}{r}} - \frac{ \e^{V} \rho M_{t}}{r}\pnth{1 - \frac{2M}{r}}^{-1/2}} \\
  \label{rhoEq-1_t} &= V_{t}M_{t} + 4\pi r^{2}\e^{V}\brkt{\rho_{t}\sqrt{1 - \frac{2M}{r}} - \frac{\rho M_{t}}{r}\pnth{1 - \frac{2M}{r}}^{-1/2}}
\end{align}
and
\begin{align}
  \notag \p_{r}V_{r} &= \p_{r}\brkt{\pnth{1 - \frac{2M}{r}}^{-1}\pnth{\frac{M}{r^{2}} + 4\pi r P}} \\
  \notag V_{rr} &= 2\pnth{1 - \frac{2M}{r}}^{-2}\pnth{\frac{M_{r}}{r} - \frac{M}{r^{2}}}\pnth{\frac{M}{r^{2}} + 4\pi r P} 
  + \pnth{1 - \frac{2M}{r}}^{-1}\pnth{\frac{M_{r}}{r^{2}} - \frac{2M}{r^{3}} + 4\pi P + 4\pi r P_{r}} \\
  \label{PREq-1_r} 
  &= \pnth{2\frac{V_{r}}{r} + \frac{1}{r^{2}}}\pnth{1 - \frac{2M}{r}}^{-1}\pnth{M_{r} - \frac{M}{r}} - \frac{V_{r}}{r} + 4\pi \pnth{1 - \frac{2M}{r}}^{-1}(2P + rP_{r})
\end{align}
Making these substitutions into (\ref{QEq-1}) yields
\begin{align}
  \notag 8\pi Q &= 
  4\pi(2P + rP_{r}) + V_{r}^{2}\pnth{1 - \frac{2M}{r}} + \frac{V_{r}}{r}\pnth{M_{r} - \frac{M}{r}}  \\
  \notag &\qquad -\frac{4\pi r}{\e^{V}}\brkt{\rho_{t}\pnth{1 - \frac{2M}{r}}^{-1/2} - \frac{\rho M_{t}}{r}\pnth{1 - \frac{2M}{r}}^{-3/2}} 
  - \frac{3M_{t}^{2}}{r^{2} \e^{2V}}\pnth{1-\frac{2M}{r}}^{-2} \displaybreak[0] \\
  \notag &= 4\pi(2P + rP_{r}) + V_{r}\pnth{\frac{M}{r^{2}} + 4\pi r P} + \frac{V_{r}}{r}\pnth{M_{r} - \frac{M}{r}}  \\
  \notag &\qquad -\frac{4\pi r}{\e^{V}}\brkt{\rho_{t}\pnth{1 - \frac{2M}{r}}^{-1/2} - \frac{\rho M_{t}}{r}\pnth{1 - \frac{2M}{r}}^{-3/2}} 
  - \frac{12\pi \rho M_{t}}{\e^{V}}\pnth{1-\frac{2M}{r}}^{-3/2} \displaybreak[0] \\
  \notag &= 4\pi(2P + rP_{r}) + 4\pi r V_{r}(P + \mu) - \frac{8\pi \rho M_{t}}{\e^{V}}\pnth{1-\frac{2M}{r}}^{-3/2} \\
  \notag &\qquad - 4\pi r \pnth{V_{r}(\mu + P) + P_{r} + \frac{2(P - Q)}{r}} + \frac{8\pi\rho M_{t}}{\e^{V}}\pnth{1 - \frac{2M}{r}}^{-3/2} \\
  \label{QEq-2} &= 8\pi Q
\end{align}
where the last two lines follow from using (\ref{divT3b}).  Thus so long as equation (\ref{divT3}) holds, equations (\ref{EDEq-1-rp}), (\ref{PREq-1-rp}), and (\ref{rhoEq-1}) imply equation (\ref{QEq-1}).  Since equations (\ref{divT3}), (\ref{EDEq-1-rp}), and (\ref{PREq-1-rp}) imply equation (\ref{rhoEq-1}) and equation (\ref{divT3}) follows from $\divx{g}T = 0$, we have that solving $\divx{g}T = 0$ 
and solving equations (\ref{EDEq-1-rp}) and (\ref{PREq-1-rp}) at each time $t$ also solves the entire Einstein equation, which was the desired result.
\end{proof}

The proof of Lemma \ref{SET-Lem-1} follows.

\begin{proof}[Proof of Lemma \ref{SET-Lem-1}]
  Recall equation (\ref{SET-Def}).  Then to compute the components of $T$, we first have that
    \begin{align}
      df &= f_{t}\, dt + f_{r}\, dr = p\e^{V}\sqrt{1 - \frac{2M}{r}}\, dt + f_{r}\, dr \\
      d\cf &= \cf_{t}\, dt + \cf_{r}\, dr = \cp\e^{V}\sqrt{1 - \frac{2M}{r}}\, dt + \cf_{r}\, dr.
    \end{align}
    Then, since $g$ is diagonal, this implies that
    \begin{align}\label{df2}
      \notag \abs{df}^{2} &= g(df,d\cf) = g\pnth{p\e^{V}\sqrt{1 - \frac{2M}{r}}\, dt + f_{r}\, dr,\, \cp\e^{V}\sqrt{1 - \frac{2M}{r}}\, dt + \cf_{r}\, dr} \\
      \notag &= \e^{2V}\pnth{1 - \frac{2M}{r}}p\cp\, g(dt,dt) + f_{r} \cf_{r}\, g(dr,dr) \\
      &= \pnth{1 - \frac{2M}{r}}\pnth{\abs{f_{r}}^{2} - \abs{p}^{2}}.
    \end{align}
    With these facts, we now compute the quantities in question and do so in the same order as they were presented.  Thus we have the following.
    \begin{align}
      \notag T(\nu_{t},\nu_{t}) &= \mu_{0}\pnth{\frac{2}{\Up^{2}} df(\nu_{t})d\cf(\nu_{t}) - \brkt{\pnth{1 - \frac{2M}{r}}\frac{\abs{f_{r}}^{2} - \abs{p}^{2}}{\Up^2} + \abs{f}^{2}}g(\nu_{t},\nu_{t})} \\
      &= \mu_{0}\pnth{\abs{f}^{2} + \pnth{1 - \frac{2M}{r}}\frac{\abs{f_{r}}^{2} + \abs{p}^{2}}{\Up^{2}}}.
    \end{align}
    Next, because $\nu_{\eta}$ is an everywhere orthonormal basis,
    \begin{align}
      \notag T(\nu_{t},\nu_{r}) &= \mu_{0}\Bigg(\frac{1}{\Up^{2}} \brkt{df(\nu_{t})d\cf(\nu_{r}) + d\cf(\nu_{t})df(\nu_{r})} 
      - \brkt{\pnth{1 - \frac{2M}{r}}\frac{\abs{f_{r}}^{2} - \abs{p}^{2}}{\Up^2} + \abs{f}^{2}}g(\nu_{t},\nu_{r})\Bigg) \\
      &= \frac{2\mu_{0}}{\Up^{2}}\pnth{1 - \frac{2M}{r}}\Re(f_{r}\cp).
    \end{align}
    \begin{align}
      \notag T(\nu_{r},\nu_{r}) &= \mu_{0}\pnth{\frac{2}{\Up^{2}} df(\nu_{r})d\cf(\nu_{r}) - \brkt{\pnth{1 - \frac{2M}{r}}\frac{\abs{f_{r}}^{2} - \abs{p}^{2}}{\Up^2} + \abs{f}^{2}}g(\nu_{r},\nu_{r})} \\
      &= \mu_{0}\pnth{-\abs{f}^{2} + \pnth{1 - \frac{2M}{r}}\frac{\abs{f_{r}}^{2} + \abs{p}^{2}}{\Up^{2}}}.
    \end{align}
    \begin{align}
      \notag T(\nu_{\theta},\nu_{\theta}) &= \mu_{0}\pnth{\frac{2}{\Up^{2}} df(\nu_{\theta})d\cf(\nu_{\theta}) - \brkt{\pnth{1 - \frac{2M}{r}}\frac{\abs{f_{r}}^{2} - \abs{p}^{2}}{\Up^2} + \abs{f}^{2}}g(\nu_{\theta},\nu_{\theta})} \\
      &= -\mu_{0}\pnth{\abs{f}^{2} + \pnth{1 - \frac{2M}{r}}\frac{\abs{f_{r}}^{2} - \abs{p}^{2}}{\Up^{2}}}.
    \end{align}
    Finally,
    \begin{align}
      \notag T(\nu_{\varphi},\nu_{\varphi}) &= \mu_{0}\pnth{\frac{2}{\Up^{2}} df(\nu_{\varphi})d\cf(\nu_{\varphi}) - \brkt{\pnth{1 - \frac{2M}{r}}\frac{\abs{f_{r}}^{2} - \abs{p}^{2}}{\Up^2} + \abs{f}^{2}}g(\nu_{\varphi},\nu_{\varphi})} \\
      &= -\mu_{0}\pnth{\abs{f}^{2} + \pnth{1 - \frac{2M}{r}}\frac{\abs{f_{r}}^{2} - \abs{p}^{2}}{\Up^{2}}}
    \end{align}
    so that $T(\nu_{\varphi},\nu_{\varphi}) = T(\nu_{\theta},\nu_{\theta})$.
\end{proof}

Finally, we prove Lemma \ref{divTEKG-Lem}.

\begin{proof}[Proof of Lemma \ref{divTEKG-Lem}]
  By equation (\ref{SETFun-def}) and Lemma \ref{SET-Lem-1}, we have that
  \begin{subequations}\label{SETEKGFun}
    \begin{align}
      \label{SETEKGFun-1} \mu &= \mu_{0}\pnth{\abs{f}^{2} + \pnth{1 - \frac{2M}{r}}\frac{\abs{f_{r}}^{2} + \abs{p}^{2}}{\Up^{2}}} \\
      \label{SETEKGFun-2} \rho &= \frac{2\mu_{0}}{\Up^2}\pnth{1 - \frac{2M}{r}}\Re(f_{r}\cp) \\
      \label{SETEKGFun-3} P &= \mu_{0}\pnth{-\abs{f}^{2} + \pnth{1 - \frac{2M}{r}}\frac{\abs{f_{r}}^{2} + \abs{p}^{2}}{\Up^{2}}} \\
      \label{SETEKGFun-4} Q &= -\mu_{0}\pnth{\abs{f}^{2} + \pnth{1 - \frac{2M}{r}}\frac{\abs{f_{r}}^{2} - \abs{p}^{2}}{\Up^{2}}}
    \end{align}
  \end{subequations}
  We know from a previous proof that $\divx{g}T=0$ is equivalent to equation (\ref{divT3}).  Moreover, by (\ref{p-Def}) and Lemma \ref{KG-Lem-1}, we have that (\ref{EKG-2}) is equivalent to (\ref{NCpde4}) and (\ref{NCpde3}).  Thus it suffices to show that given (\ref{SETEKGFun}), equations (\ref{NCpde4}) and (\ref{NCpde3}) imply equation (\ref{divT3}).

  We will start with (\ref{divT3a}).  On one hand, the left hand side is equal to differentiating (\ref{SETEKGFun-1}) with respect to $t$.  \begin{subequations}\label{divTEKG2-a}
    \begin{align}
      \notag \mu_{t} &= \p_{t}\brkt{\mu_{0}\pnth{\abs{f}^{2} + \pnth{1 - \frac{2M}{r}}\frac{\abs{f_{r}}^{2} + \abs{p}^{2}}{\Up^{2}}}} \\
      \notag &= \mu_{0}\brkt{f_{t}\cf + f \cf_{t} - \frac{2M_{t}}{r}\pnth{\frac{\abs{f_{r}}^{2} + \abs{p}^{2}}{\Up^{2}}} + \pnth{1 - \frac{2M}{r}}\frac{f_{rt}\cf_{r} + f_{r}\cf_{rt} + p_{t}\cp + p\cp_{t}}{\Up^{2}}} \displaybreak[0]\\
      \notag \mu_{t} &= \mu_{0}\Bigg[\cp \pnth{f \e^{V}\sqrt{1 - \frac{2M}{r}} + \frac{p_{t}}{\Up^{2}}\pnth{1 - \frac{2M}{r}}} \\
      \notag &\qquad + p \pnth{\cf \e^{V}\sqrt{1 - \frac{2M}{r}} + \frac{\cp_{t}}{\Up^{2}}\pnth{1 - \frac{2M}{r}}} - \frac{2M_{t}}{r}\pnth{\frac{\abs{f_{r}}^{2} + \abs{p}^{2}}{\Up^{2}}}  \\
      \label{divTEKG2-aLHS} &\qquad + \pnth{1 - \frac{2M}{r}}\frac{\cf_{r}\p_{r}\pnth{p \e^{V}\sqrt{1 - \frac{2M}{r}}} + f_{r}\p_{r}\pnth{\cp \e^{V}\sqrt{1 - \frac{2M}{r}}}}{\Up^{2}} \Bigg]
    \end{align}
    On the other hand, we substitute (\ref{SETEKGFun}) into (\ref{divT3a}).
    \begin{align}
      \notag \mu_{t} &= 
      2\e^{V}\sqrt{1 - \frac{2M}{r}}\pnth{\frac{2\mu_{0}}{\Up^2}\pnth{1 - \frac{2M}{r}}\Re(f_{r}\cp)}\pnth{V_{r} + \frac{1}{r}} \\
      \notag &\qquad + \p_{r}\brkt{\frac{2\mu_{0}}{\Up^2}\pnth{1 - \frac{2M}{r}}\Re(f_{r}\cp)}\e^{V}\sqrt{1 - \frac{2M}{r}} - \frac{2M_{t}\mu_{0}}{r}\pnth{\frac{\abs{f_{r}}^{2} + \abs{p}^{2}}{\Up^{2}}} \displaybreak[0] \\
      \notag &= \cp \frac{\mu_{0}}{\Up^{2}}\pnth{1 - \frac{2M}{r}} \brkt{f_{r}V_{r}\e^{V}\sqrt{1 - \frac{2M}{r}} + \frac{2f_{r}\e^{V}}{r}\sqrt{1 - \frac{2M}{r}}} \\
      \notag & \qquad + p \frac{\mu_{0}}{\Up^{2}}\pnth{1 - \frac{2M}{r}} \brkt{\cf_{r}V_{r}\e^{V}\sqrt{1 - \frac{2M}{r}} + \frac{2\cf_{r}\e^{V}}{r}\sqrt{1 - \frac{2M}{r}}} \\
      \notag & \qquad + \frac{\mu_{0}V_{r}\e^{V}}{\Up^{2}}\pnth{1 - \frac{2M}{r}}^{3/2}\pnth{f_{r}\cp + p \cf_{r}} - \frac{2M_{t}\mu_{0}}{r}\pnth{\frac{\abs{f_{r}}^{2} + \abs{p}^{2}}{\Up^{2}}} \\
      \notag & \qquad + \frac{\mu_{0}\e^{V}}{\Up^2}\sqrt{1 - \frac{2M}{r}}\Bigg[(f_{r}\cp + p \cf_{r})\p_{r}\pnth{1 - \frac{2M}{r}} 
      + \pnth{1 - \frac{2M}{r}}(f_{rr}\cp + p \cf_{rr} + f_{r}\cp_{r} + p_{r} \cf_{r})\Bigg] \displaybreak[0] \\
      \notag \mu_{t} &= \cp \frac{\mu_{0}}{\Up^{2}}\pnth{1 - \frac{2M}{r}} \brkt{\p_{r}\pnth{f_{r}\e^{V}\sqrt{1 - \frac{2M}{r}}} + \frac{2f_{r}\e^{V}}{r}\sqrt{1 - \frac{2M}{r}}} \\
      \notag & \qquad + p \frac{\mu_{0}}{\Up^{2}}\pnth{1 - \frac{2M}{r}} \brkt{\p_{r}\pnth{\cf_{r}\e^{V}\sqrt{1 - \frac{2M}{r}}} + \frac{2\cf_{r}\e^{V}}{r}\sqrt{1 - \frac{2M}{r}}} 
      - \frac{2M_{t}\mu_{0}}{r}\pnth{\frac{\abs{f_{r}}^{2} + \abs{p}^{2}}{\Up^{2}}} \\
      \label{divTEKG2-aRHS} & \qquad + \frac{\mu_{0}}{\Up^2}\pnth{1 - \frac{2M}{r}}\brkt{\cf_{r}\p_{r}\pnth{p\e^{V}\sqrt{1 - \frac{2M}{r}}} + f_{r}\p_{r}\pnth{\cp\e^{V}\sqrt{1 - \frac{2M}{r}}}}
    \end{align}
    Setting (\ref{divTEKG2-aLHS}) equal to (\ref{divTEKG2-aRHS}), we obtain
    \begin{align}
      \notag \cp &\pnth{\Up^{2} f \e^{V}\pnth{1 - \frac{2M}{r}}^{-1/2} + p_{t}} + p \pnth{\Up^{2}\cf \e^{V}\pnth{1 - \frac{2M}{r}}^{-1/2} + \cp_{t}} \\
      \notag &= \cp \brkt{\p_{r}\pnth{f_{r}\e^{V}\sqrt{1 - \frac{2M}{r}}} + \frac{2f_{r}\e^{V}}{r}\sqrt{1 - \frac{2M}{r}}} 
      + p \brkt{\p_{r}\pnth{\cf_{r}\e^{V}\sqrt{1 - \frac{2M}{r}}} + \frac{2\cf_{r}\e^{V}}{r}\sqrt{1 - \frac{2M}{r}}} \displaybreak[0]\\
      \notag \displaybreak[0] \\
      \notag 0 &= \cp\brkt{- p_{t} + \e^{V}\pnth{-\Up^{2} f \pnth{1 - \frac{2M}{r}}^{-1/2} + \frac{2f_{r}}{r}\sqrt{1 - \frac{2M}{r}}} + \p_{r}\pnth{f_{r}\e^{V}\sqrt{1 - \frac{2M}{r}}}} \\
      \label{divTEKG2-afin} &\ \ \, + p\brkt{- \cp_{t} + \e^{V}\pnth{-\Up^{2} \cf \pnth{1 - \frac{2M}{r}}^{-1/2} + \frac{2\cf_{r}}{r}\sqrt{1 - \frac{2M}{r}}} + \p_{r}\pnth{\cf_{r}\e^{V}\sqrt{1 - \frac{2M}{r}}}}
    \end{align}
  \end{subequations}
  Note that, since $V$ and $M$ are real valued, the second term above is the complex conjugate of the first term.  Then by (\ref{NCpde4}) equation (\ref{divTEKG2-afin}) holds.

  Next, we consider (\ref{divT3b}).  We first differentiate (\ref{SETEKGFun-2}) with respect to $t$. 
  Using equation (\ref{NCpde3}), this yields
  \begin{subequations}\label{divTEKG2-b}
    \begin{align}
      \notag \rho_{t} &= \p_{t}\brkt{\frac{2\mu_{0}}{\Up^2}\pnth{1 - \frac{2M}{r}}\Re(f_{r}\cp)} \\
      \notag &= \frac{\mu_{0}}{\Up^{2}}\brkt{-\frac{4M_{t}}{r}\Re(f_{r}\cp) + \pnth{1 - \frac{2M}{r}}\pnth{f_{rt}\cp + \cf_{rt}p + f_{r}\cp_{t} + \cf_{r}p_{t}}} \displaybreak[0]\\
      \label{divTEKG2-bLHS} &= \frac{\mu_{0}}{\Up^{2}}\Bigg\{-\frac{4M_{t}}{r}\Re(f_{r}\cp) + \pnth{1 - \frac{2M}{r}}\Bigg[\p_{r}\pnth{\abs{p}^{2}\e^{V}\sqrt{1 - \frac{2M}{r}}} 
      + \abs{p}^{2}\p_{r}\pnth{\e^{V}\sqrt{1 - \frac{2M}{r}}} + f_{r}\cp_{t} + \cf_{r}p_{t}\Bigg]\Bigg\}
    \end{align}
    On the other hand, we substitute (\ref{SETEKGFun}) into (\ref{divT3b}), to obtain  
    \begin{align}
      \notag \rho_{t} &= 
      \e^{V}\sqrt{1 - \frac{2M}{r}}\Bigg[\frac{2\mu_{0}V_{r}}{\Up^{2}}\pnth{1 - \frac{2M}{r}}(\abs{f_{r}}^{2} + \abs{p}^{2}) \\
      \notag &\ \ + \mu_{0}\p_{r}\pnth{-\abs{f}^{2} + \pnth{1 - \frac{2M}{r}}\frac{\abs{f_{r}}^{2} + \abs{p}^{2}}{\Up^{2}}} + \frac{4\mu_{0}\abs{f_{r}}^{2}}{r\Up^{2}}\pnth{1 - \frac{2M}{r}}\Bigg] - \frac{4\mu_{0}M_{t}}{r\Up^{2}}\Re(f_{r}\cp) \displaybreak[0] \\
      \notag &= \frac{\mu_{0}}{\Up^{2}}\pnth{1 - \frac{2M}{r}}\Bigg\{\cf_{r}\Bigg[\p_{r}\pnth{f_{r}\e^{V}\sqrt{1 - \frac{2M}{r}}} 
      +\e^{V} \pnth{- \Up^{2}f\pnth{1 - \frac{2M}{r}}^{-1/2} + \frac{2 f_{r}}{r}\sqrt{1 - \frac{2M}{r}}}\Bigg] \\
      \notag &\qquad + f_{r}\brkt{\p_{r}\pnth{\cf_{r}\e^{V}\sqrt{1 - \frac{2M}{r}}} +\e^{V}\pnth{-\Up^{2}\cf\pnth{1 - \frac{2M}{r}}^{-1/2} + \frac{2 \cf_{r}}{r}\sqrt{1 - \frac{2M}{r}}}} \\
      \label{divTEKG2-bRHS} &\qquad + \abs{p}^{2}\p_{r}\pnth{\e^{V}\sqrt{1 - \frac{2M}{r}}} + \p_{r}\pnth{\abs{p}^{2}\e^{V}\sqrt{1 - \frac{2M}{r}}} \Bigg\} - \frac{4\mu_{0}M_{t}}{r\Up^{2}}\Re(f_{r}\cp)
    \end{align}
    Equating (\ref{divTEKG2-bLHS}) and (\ref{divTEKG2-bRHS}) yields
    \begin{align}
      \notag 0 &= \cf_{r}\Bigg[- p_{t} + \e^{V}\pnth{-\Up^{2} f \pnth{1 - \frac{2M}{r}}^{-1/2} + \frac{2f_{r}}{r}\sqrt{1 - \frac{2M}{r}}} 
      + \p_{r}\pnth{f_{r}\e^{V}\sqrt{1 - \frac{2M}{r}}}\Bigg] \\
      \label{divTEKG2-bfin} &\qquad + f_{r}\Bigg[- \cp_{t} + \e^{V}\pnth{-\Up^{2} \cf \pnth{1 - \frac{2M}{r}}^{-1/2} + \frac{2\cf_{r}}{r}\sqrt{1 - \frac{2M}{r}}} 
      + \p_{r}\pnth{\cf_{r}\e^{V}\sqrt{1 - \frac{2M}{r}}}\Bigg]
    \end{align}
  \end{subequations}
  As before, since $M$ and $V$ are real valued, the second term above is the complex conjugate of the first term.  Then by equation (\ref{NCpde4}), equation (\ref{divTEKG2-bfin}) holds.  Then equations (\ref{NCpde4}) and (\ref{NCpde3}) imply equation (\ref{divT3}) holds, which completes the proof.
\end{proof}

%
%

\begin{acknowledgements}
The author would like to thank Hubert Bray and Andrew Goetz for several extremely useful discussions on this topic.  He also gratefully acknowledges the support of the National Science Foundation grant \#~DMS-1007063.
\end{acknowledgements}

\bibliographystyle{spmpsci}      
\bibliography{./References}   

%
%

\end{document}